\documentclass[fontsize=10pt, a4paper, oneside, headings=standardclasses, DIV=12]{scrartcl}
\usepackage[utf8]{inputenc}
\usepackage[T1]{fontenc}
\usepackage{lmodern}
\usepackage{microtype}
\usepackage{titling}

\usepackage{braket}
\usepackage{graphicx}
\usepackage{multicol,multirow}
\usepackage{amsmath,amssymb,amsfonts}
\usepackage{mathrsfs}
\usepackage{amsthm}
\usepackage{rotating}
\usepackage{appendix}
\usepackage{csquotes}
\usepackage[T1]{fontenc}
\usepackage{newtxtext}
\usepackage{newtxmath}
\usepackage{textcomp}
\usepackage{xcolor}
\usepackage{xparse}
\usepackage{xspace}
\usepackage[colorlinks,allcolors=blue]{hyperref}
\usepackage[style=alphabetic, maxbibnames=99]{biblatex}
\usepackage{bbm}
\usepackage{bbold}
\usepackage{aligned-overset}

\usepackage{bbding}
\usepackage{stmaryrd}
\usepackage{enumitem}

\usepackage{newunicodechar}
\newunicodechar{∗}{\textasteriskcentered}
\DeclareUnicodeCharacter{0301}{*************************************}

\usepackage{thmtools}
\newtheorem{theorem}{Theorem}[section]
\newtheorem{proposition}[theorem]{Proposition}
\newtheorem{lemma}[theorem]{Lemma}
\newtheorem{definition}[theorem]{Definition}
\theoremstyle{definition}
\newtheorem{remark}[theorem]{Remark}

\numberwithin{equation}{section}

\usepackage{environ}

\makeatletter
\long\def\@secondofthree#1#2#3{#2}
\long\def\@thirdoffour#1#2#3#4{#3}

\NewEnviron{restatethis}[2]{%
  \begin{#1}%
      \BODY%
  \end{#1}%
  \protected@write\@auxout{}{%
    \string\@restatetheorem{#1}{#2}{\csname the#1\endcsname}{\detokenize\expandafter{\BODY}}%
  }%
}

\@ifpackageloaded{thmtools}{%
    \def\restatethm@getthmcountercsname#1{\def\thethmcsname{#1}}%
}{%
    \@ifpackageloaded{ntheorem}{%
        \def\restatethm@getthmcountercsname#1{%
            \def\thethmcsname{\expandafter\expandafter\expandafter\restatethm@ntheorem@getthmcountercsname@helper\csname mkheader@#1\endcsname}}%
        \def\restatethm@ntheorem@getthmcountercsname@helper#1\@thm#2#3#4{#3}
    }{%
        \@ifpackageloaded{amsthm}{%
            \def\restatethm@getthmcountercsname#1{\edef\thethmcsname{\expandafter\expandafter\expandafter\@thirdoffour\csname#1\endcsname}}%
        }{%
            \def\restatethm@getthmcountercsname#1{\edef\thethmcsname{\expandafter\expandafter\expandafter\@secondofthree\csname#1\endcsname}}%
        }%
    }%
}

\newcommand{\@restatetheorem}[4]{%
  \expandafter\gdef\csname restatethis@#2\endcsname{%
    \begingroup
    \restatethm@getthmcountercsname{#1}
    \expandafter\def\csname the\thethmcsname\endcsname{#3}%
    \begin{#1}#4\end{#1}%
    \endgroup
  }%
}

\newcommand{\restate}[1]{\csname restatethis@#1\endcsname} 
\makeatother

\theoremstyle{definition}
\newtheorem{corollary}[theorem]{Corollary}

\DeclareMathOperator{\argmax}{argmax}

\DeclareMathOperator{\Tr}{Tr}

\DeclareMathOperator{\id}{id}
\newcommand{\ketbra}[2]{\left|#1\left\rangle\right\langle#2\right|}
\newcommand{\1}{\ensuremath{\mathbbm{1}}}

\newcommand{\E}{\mathcal{E}}
\newcommand{\F}{\mathcal{F}}
\newcommand{\D}{\mathbf{D}}
\renewcommand{\P}{\mathcal{P}}

\newcommand{\reals}{\mathbb{R}}

\newcommand{\naturals}{\mathbb{N}}

\newcommand{\br}[1]{\left(#1\right)}

\newcommand{\HS}{\mathcal{H}}

\NewDocumentCommand{\BO}{o}{\mathcal{B}\br{\IfValueTF{#1}{#1}{\HS}}}

\NewDocumentCommand{\DM}{o}{\mathcal{D}\br{\IfValueTF{#1}{#1}{\HS}}}

\NewDocumentCommand{\pos}{o}{\mathscr{P}\!\br{\IfValueTF{#1}{#1}{\HS}}}

\newcommand{\renyi}{R\'enyi\xspace}

\DeclarePairedDelimiter\abs{\lvert}{\rvert}%
\DeclarePairedDelimiter\norm{\lVert}{\rVert}%
\makeatletter
\let\oldabs\abs
\def\abs{\@ifstar{\oldabs}{\oldabs*}}
\let\oldnorm\norm
\def\norm{\@ifstar{\oldnorm}{\oldnorm*}}
\newcommand{\cB}{\mathcal{B}}
\newcommand{\Dg}{\widehat{D}_{\alpha}}
\newcommand{\Sg}{\widehat{S}_{\alpha}}
\newcommand{\Gb}{\bar{\Gamma}_n}
\newcommand{\n}{^{\otimes n}}

\NewDocumentCommand{\cptp}{o o}{\mathrm{CPTP}\br{\IfValueTF{#1}{#1}{\DM[\HS_A]} \to \IfValueTF{#1}{#1}{\DM[\HS_B]}}}

\newcommand{\bbcomm}[1]{\textcolor{red}{[BB: {#1}]}}

\DeclareNameFormat{initgiven-family}{\usebibmacro{name:given-family}
      {\namepartfamily}
      {\namepartgiveni}
      {\namepartprefix}
      {\namepartsuffix}
  \usebibmacro{name:andothers}}
  
\newbibmacro*{cite:abscite}{%
\mkbibbrackets{\printtext[bibhyperref]{%
      \printnames[initgiven-family]{labelname}%
      \setunit{\addcomma\space}%
      \iffieldundef{journaltitle}{\iffieldundef{eprint}{\printfield{doi}}{\clearfield{eprintclass}\printfield{eprint}}}{\printfield{doi}}%
      }}}

\DeclareCiteCommand{\abscite}
  {\usebibmacro{prenote}}%
  {\usebibmacro{citeindex}%
   \usebibmacro{cite:abscite}}
  {\multicitedelim}
  {\usebibmacro{postnote}}

\AtEveryBibitem{
	\clearfield{note}
	\clearfield{issn}
	\clearfield{urlyear}
	\clearfield{urlmonth}
	\clearfield{eprintclass}
	\iffieldundef{journaltitle}{}{\clearfield{version}}
	\iffieldundef{journaltitle}{\iffieldundef{eprint}{}{\clearfield{doi}}}{}
}
\renewbibmacro*{url}{%
	\iffieldundef{doi}{\iffieldundef{eprint}{%
			\printfield{url}%
		}{}}{}
}

\renewcommand*{\multicitedelim}{\addcomma\space}

\title{Infinite Dimensional Asymmetric Quantum Channel Discrimination}
\author{Bjarne Bergh\thanks{Department of Applied Mathematics and Theoretical Physics, University of Cambridge, United Kingdom} \thanks{Corresponding author: bb536@cam.ac.uk} \and Jan Kochanowski\thanks{Max-Planck-Institute of Quantum Optics, Garching, Germany} \and  Robert Salzmann\thanksmark{1} \and Nilanjana Datta\thanksmark{1}}

\newcommand{\kb}[1]{|#1\rangle\!\langle#1|} 

\begin{document}


\maketitle

\begin{abstract}
We study asymmetric binary channel discrimination, for qantum channels acting on separable Hilbert spaces. We establish quantum Stein's lemma for channels for both adaptive and parallel strategies, and show that under finiteness of the geometric R\'enyi divergence between the two channels for some $\alpha > 1$, adaptive strategies offer no asymptotic advantage over parallel ones. One major step in our argument is to demonstrate that the geometric R\'enyi divergence satisfies a chain rule and is additive for channels also in infinite dimensions. These results may be of independent interest. Furthermore, we not only show asymptotic equivalence of parallel and adaptive strategies, but explicitly construct a parallel strategy which approximates a given adaptive $n$-shot strategy, and give an explicit bound on the difference between the discrimination errors for these two strategies. This extends the finite dimensional result from ~\abscite{art:Bjarne}. Finally, this also allows us to conclude, that the chain rule for the Umegaki relative entropy in infinite dimensions, recently shown in~\abscite{fawzi_asymptotic_2023} given finiteness of the max divergence between the two channels, also holds under the weaker condition of finiteness of the geometric \renyi divergence. We give explicit examples of channels which show that these two finiteness conditions are not equivalent. 
\end{abstract}

\tableofcontents


\section{Introduction}

The task of quantum channel discrimination can be described as follows: Given an unknown quantum channel as a black box and the side information that it is one of two possible channels, the task is to determine the channel's identity \cite{chiribella_memory_2008,duan_perfect_2009,hayashi_discrimination_2009,harrow_adaptive_2010}.
This task boils down to finding the optimal quantum state to send as an input to the channel and then also the best measurement to perform on the output of the channel. If we are given access to $n$ copies of the channel (i.e.~we are given $n$ identical black boxes, each of which can be used once); then there are different strategies (sometimes also called protocols) in which we could set up our decision experiment -- the so-called {\em{parallel}} and {\em{adaptive}} strategies.
In a parallel strategy one prepares a joint state, usually entangled between the input systems of all the $n$ copies of the channel and an additional ancillary (or reference or memory) system. This state is then fed as an input to all the $n$ channels at once (with the state of the ancillary system being left undisturbed). Finally, a binary positive operator-valued measure (POVM) is performed on the joint ouput state of the channels and the ancillary system in order to arrive at a decision for the channel's identity.
    In an adaptive strategy, on the other hand, one prepares an input state for a single copy of the channel (again possibly entangled with a ancillary system) which is fed into the first instance of the channel, with the state of the ancillary system being left undisturbed. The input to the next use of the channel is then chosen depending on the output of the first channel and the state of our ancillary system. This is done, most generally, by subjecting the latter to an arbitrary quantum operation (or channel), which we call a preparation operation. This step is repeated for each successive use of the channel until all the $n$ black-boxes have been used.
	 Then, a binary POVM is performed on the joint state of the output of the last instance of the channel and the ancillary system. See \autoref{fig:plot_adaptive} for a depiction of an adaptive strategy. Adaptive strategies are also sometimes called sequential, which is however not to be confused with the setting of sequential hypothesis testing \cite{martinez_vargas_quantum_2021, li_optimal_2022, li_sequential_2022}, where samples (i.e. states or channels) can be requested one by one. 
	 
	 One particularly interesting question is whether and to what degree adaptive strategies give an advantage over parallel ones. Note that any parallel strategy can be written as an adaptive strategy by taking all but one channel input as part of the ancillary system, and then choosing each preparation operation such that it extracts the next part of the joint input state for the next channel use and replaces it by the output of the previous channel use. However, the converse is not true and hence adaptive strategies are more general. Parallel strategies are conceptually a lot simpler than adaptive ones -- aside from the measurement, everything is specified just by the joint input state -- in contrast to adaptive strategies, in which after each channel use we may perform an arbitrary quantum operation to prepare the input to the next use of the channel. It is thus interesting to determine to what degree parallel strategies can still be optimal. This problem has been recently studied in finite-dimensional settings. There, it has been shown that in certain cases adaptive strategies can give an advantage over parallel ones. Specifically, in~\cite{harrow_adaptive_2010} the authors constructed an example in which an adaptive strategy with only two channel uses could be used to discriminate the channels with certainty, which is not possible with a parallel strategy, even if arbitrarily many channel uses are allowed.

	Interestingly, \emph{asymptotically}, there are multiple known cases in which adaptive strategies give no advantage over parallel ones, i.e.~the optimal exponential decay rate of the error probability per channel use is the same in the asymptotic limit. For example, this is the case both in the symmetric and asymmetric settings when the channels are classical \cite{hayashi_discrimination_2009} or classical-quantum \cite{wilde_amortized_2020,salek_when_2021}. For arbitrary quantum channels \emph{in finite dimensions}, the recently shown chain rule for the quantum relative entropy \cite{fang_chain_2020} and the characterization of asymmetric channel discrimination in terms of amortized relative entropy \cite{wilde_amortized_2020,wang_resource_2019}, also implies that in the asymmetric setting, in the asymptotic limit where we also require the type~I error to vanish, adaptive strategies give no advantage over parallel ones (i.e. the optimal asymptotic exponential decay rate of the type~II error per channel use is the same for parallel and adaptive strategies). This is in contrast to the symmetric setting in which the example of \cite{salek_when_2021} shows that there always is an advantage of adaptive strategies, also in the asymptotic limit. Furthermore, \cite{art:Bjarne} have shown explicitly how any adaptive strategy can be converted to a parallel strategy bounding the additional error rate that occurs because of this conversion, and showing how this additional error rate vanishes asymptotically, hence giving an operational version of the asymptotic equivalence of adaptive and parallel strategies.
	
	In this paper, our aim is to study quantum channel discrimination in the infinite dimensional setting. Infinite-dimensional, continuous variable (CV) quantum systems are of huge technological and experimental relevance in quantum information.  
	This is because quantum optics, which utilizes continuous quadrature amplitudes of the quantized electromagnetic field, has been shown to be a promising platform for efficient implementation of the essential steps in quantum communication protocols, namely, preparation, (unitary) manipulation, and measurement of (entangled) quantum states, (see e.g.~\cite{art:BraunsteinReview} for a review). Examples of such systems include collections of electromagnetic modes travelling along an optical fibre, and massive
	harmonic oscillators. Infinite dimensional quantum channels which are of particular relevance include bosonic Gaussian channels (see e.g.~\cite{art:Saikat-thesis} and references therein) and bosonic dephasing channels~(see e.g.~\cite{art:lami2022exact}). The study of quantum information in infinite dimensions has found applications in quantum communication \cite{art:Ref8-PhysRevA.63.032312,art:Ref9-Wolf-2007,art:Ref10-Takeoka-2014,art:Ref11-Pirandola-2017,art:Ref12-Wilde-2017,art:Ref13-Rosati-2018, art:lami2022exact}, the most noteworthy of them being the experimental realization of quantum teleportation for optical fields. Other applications include quantum computing and quantum error correction \cite{art:Ref20-PhysRevA.64.012310,art:Ref21-Mirrahimi-2014,art:Ref22, art:Ref23-PhysRevX.6.031006, art:Ref24-Guillaud-2019}, quantum simulations \cite{art:Ref19-Flurin-2017} and quantum
	sensing \cite{art:Ref15-96fcf8b716b24c65bf698d2f11d9a9fc,art:Ref16-PhysRevLett.121.160502, art:Ref17-PhysRevLett.86.5870, art:Ref18-McCormick-2019}

	To study channel discrimination in infinite dimensions, our approach will be to follow the path of \cite{art:Bjarne} and show that a similar $n$ shot result can also be obtained in infinite dimensions (\autoref{thm:main}). To achieve this, we have to extend multiple parts of the finite dimensional quantum Shannon theory framework that has been established over the years to separable Hilbert spaces, which we consider to be the main contribution of this paper. Amongst our extensions is the chain rule for the geometric \renyi entropy in infinite dimensions, and some various other entropy inequalities mostly involving the relation between one-shot quantities and \renyi entropies. 
	
	As a further application of this framework and techniques, we get asymptotic expressions for the asymmetric error exponent (i.e. quantum Stein's lemma) in separable Hilbert spaces.	By taking suitable limits, our one-shot result then allows us to conclude the asymptotic equivalence of adaptive and parallel strategies in terms of discrimination power also in infinite dimensions, given some finiteness condition on the geometric \renyi divergence. We comment on this finiteness condition (which is not needed to show equivalence of adaptive and parallel strategies in finite dimensions) in \autoref{sec:finiteness_cond}

\section{Notation and Preliminaries}
\subsection{Notation}
In this paper $\mathcal{H},\mathcal{K}$ denote complex (in general infinite dimensional) separable Hilbert spaces. $\cB(\mathcal{H})$ denotes the Banach space of bounded operators over $\mathcal{H}$ w.r.t the operator norm $\|\cdot\|$ and $\1$ will denote the identity operator in $\cB(\mathcal{H})$. $\cB_2(\mathcal{H})$ and $\cB_1(\mathcal{H})$ denote the Banach spaces of Hilbert-Schmidt and Trace-class operators on $\mathcal{H}$ with their usual norms $\|\cdot\|_2,\|\cdot\|_1$, respectively. $\P(\HS)$ denotes the set of positive trace-class operators on $\HS$ and $\mathcal{D}(\mathcal{H})$ denotes the set of density operators on $\mathcal{H}$, i.e. the positive operators of trace 1. We use $^*$ to denote the adjoint of an operator. For $A\in \cB(\mathcal{H})$, $A^*$ is the adjoint w.r.t\ the inner product of $\HS$, while for any super-operator $S \in \cB(\HS) \to \cB(\HS)$ we write $S^*$ for the adjoint with respect to the Hilbert-Schmidt inner product. The spectrum of an operator $A$ in $\cB(\mathcal{H})$ is denoted as $\sigma(A)$.
For self-adjoint operators $A\in \cB(\mathcal{H})$ their positive and negative parts are denoted as
$A_+$ and $ A_-$, respectively and the projections onto them with $\{A\}_+,\{A\}_-$ respectively. In a similar vain will $\{A=0\}$ denote the projections onto the kernel of $A$ and we we may  write $\{A\geq B\}:=\{A-B\}_+$ as the projector onto the support of the positive part of $(A-B)$, when $A,B\in \cB(\mathcal{H})$ are both positive. In this paper, $\log$ denotes the logarithm to base 2. \\
A quantum channel $\Lambda$ is defined to be a linear, completely positive, and trace preserving (CPTP) map $\P(\mathcal{H}_A)\to\P(\mathcal{K}_B)$, where $\mathcal{H}_A,\mathcal{K}_B$ are separable Hilbert spaces. The indices $A,B$ label the physical systems described by the Hilbert spaces  $\mathcal{H}_A,\mathcal{K}_B$, respectively, and we may simply write $\Lambda\equiv\Lambda_{A\to B}$. To indicate $\rho\in\mathcal{D}(\mathcal{H}_A)$, we may simply write $\rho_A$. The identity channel $\id:\P(\mathcal{H}_A)\to\P(\mathcal{H}_{A^\prime})$ maps $\rho_A\mapsto\rho_{A^\prime}$, when $\mathcal{H}_{A}$ is isomorphic to $\mathcal{H}_{A^\prime}$. We will often drop such identity channels when it is clear what systems a channel is acting on i.e. $(\Lambda_{A\to B}\otimes \id_{A`\to B`})(\rho_{AA`})\equiv \Lambda_{A\to B}(\rho_{AA`})\equiv\Lambda(\rho_{AA`})$. In this paper all POVMs that occur are finite (or even just binary). A finite POVM (positive operator valued measure) is a finite set of bounded and positive operators $\{M_i\}_i$ such that $\sum_i M_i = \1$. By measuring a state $\rho$ in a POVM $\{M_i\}$, we mean that one obtains outcome $i$ with probability $\Tr(\rho M_i)$. Occasionally we will also associate the quantum-classical chanel $\mathcal{M}: \P(\HS) \to \P[\HS_I]$ to the POVM $\{M_i\}$, where $\mathcal{M}(\rho) = \sum_i \Tr(\rho M_i) \ketbra{i}{i}$ maps the state to the classical probability distribution of measurement outcomes (and the $\ket{i}$ form a basis of $\HS_I$ whose dimension equals the number of POVM elements). The most general measurement for a binary state or channel discrimination can be described by a binary POVM $\{F,\1 -F\}$ where $0\le F\le \1.$ 

\subsection{Quantum Information measures}
In the following we will give a quick overview of what is known about certain quantum divergences (interchangeably called relative entropies) and distance measures in infinite dimensions. 

\subsubsection{Fidelity and Sine-distance}
The \textit{quantum} or \textit{Uhlmann fidelity} \cite{art:UHLMANN1976273} between two quantum states $\rho,\sigma \in\mathcal{D}(\mathcal{H})$, where $\mathcal{H}$ denotes a separable Hilbert space, is given by
\begin{align}
    F(\rho,\sigma):=\|\sqrt{\rho}\sqrt{\sigma}\|_1^2.
\end{align}
It is symmetric in its arguments and
can only increase under the action of a quantum channel $\Lambda$:
$F(\rho,\sigma) \leq F(\Lambda(\rho),\Lambda(\sigma))$, see e.g. \cite{art:PrinciplesOfQuantumCommunicationTheory2020}.
The \textit{Sine-distance} (or \textit{purified distance}) \cite{rastegin_relative_2002,rastegin_lower_2003,gilchrist_distance_2005,rastegin_sine_2006} between two states $\rho,\sigma$ is defined as
\begin{equation}
P(\rho,\sigma):=\sqrt{1-F(\rho,\sigma)}.
\end{equation}
It is a metric on $\mathcal{D}(\mathcal{H})$, so in particular, it satisfies the triangle inequality. It also satisfies a so-called {\em{data processing inequality}} (DPI): $P(\Lambda(\rho),\Lambda(\sigma))\leq P(\rho,\sigma)$ for any quantum channel $\Lambda$. Proofs of these properties follow from elementary properties of the fidelity and Uhlmann's theorem, and thus also hold in separable Hilbert spaces, see e.g. \cite{art:PrinciplesOfQuantumCommunicationTheory2020}.

Recall the well known relations between the Uhlmann Fidelity and the trace-distance between two states.
\begin{lemma}[Fuchs-Van-de-Graaf Inequality]\label{lem:FuchsVanDeGraaf}\cite{art:PrinciplesOfQuantumCommunicationTheory2020}
For two states $\rho,\sigma\in\mathcal{D}(\mathcal{H})$ on a separable Hilbert space $\mathcal{H}$ the following inequality holds:
\begin{align}
    1-\sqrt{F(\rho,\sigma)}\leq \frac{1}{2}\|\rho-\sigma\|_1\leq\sqrt{1-F(\rho,\sigma)}=P(\rho,\sigma).
\end{align}
\end{lemma}
The proof of this Lemma essentially relies only on Uhlmann's theorem and thus holds true in arbitrary separable Hilbert spaces.

\subsubsection{Gentle Measurement Lemma}
The {\em{Gentle Measurement Lemma}}~\cite{art:winter1999coding} states that if a POVM measurement on a state almost always yields one specific outcome, the post-measurement state of this outcome is close to the original state.
\begin{lemma}[Gentle measurement Lemma]\cite{art:winter1999coding}
Let $\epsilon\in[0,1]$ and $\rho\in\mathcal{D}(\mathcal{H})$ be a quantum state on a separable Hilbert space $\mathcal{H}$. Let $F\in \cB(\mathcal{H})$ be such that $0\leq F\leq\1$ and $\Tr[F\rho]\geq 1-\epsilon$. Then the post-measurement state
\begin{align*}
    \rho^\prime = \frac{\sqrt{F}\rho\sqrt{F}}{\Tr[F\rho]}
\end{align*} satisfies
$P(\rho,\rho^\prime)\leq \sqrt{\epsilon}$.
\label{lem:GentleMeasurement}
\end{lemma}
As this lemma can again be proven using Uhlmann's theorem it holds also in infinite dimensions.

\subsubsection{The relative modular operator}\label{sec:rel_mod_op}
The formalism of relative modular operators, $\Delta_{\rho,\sigma}$, provides a convenient way to define certain relative entropies in von Neumann algebras and thus, manifestly, in infinite dimensional Hilbert spaces. Its formalism also proves convenient for deriving inequalities for the Umegaki relative entropy~\cite{art:UmegakiRElEntOriginal10.2996/kmj/1138844604} and the Petz-R\'enyi relative entropies~\cite{petz_quasi-entropies_1986}. For its definition in general von Neumann algebras see e.g.~\cite{art:QuantumFDivergencesinVonNeumannAlgebrasI_Hiai_2018, art:Araki1975RelativeEO}. Here we give a more explicit definition for the infinite dimensional Hilbert space setting that we are considering from \cite{art:ArakiRelModOp1977}. 
Let $\rho,\sigma\in\cB_1(\mathcal{H})$, $\rho, \sigma \geq 0$ and define the operator $S_{\rho,\sigma}:D(S)\to \cB_2(\mathcal{H})$ as \cite{art:QuantumFDivergencesViaNussbaumSzkola,art:QuantumFDivergencesinVonNeumannAlgebrasI_Hiai_2018, art:ArakiRelModOp1977}
\begin{equation}
    S_{\rho,\sigma}:(A\sqrt{\sigma}+B\{\sigma=0\}) \mapsto\{\sigma>0\}A^*\sqrt{\rho}
\end{equation}
where 
\begin{equation*}
    D(S_{\rho,\sigma}):=\{A\sqrt{\sigma}|A\in \cB(\mathcal{H})\}+\{B\{\sigma=0\}|B\in \cB_2(\mathcal{H})\}\subset \cB_2(\mathcal{H}).
\end{equation*} 
It is shown in \cite{art:ArakiRelModOp1977} that $S_{\rho,\sigma}$ is a densely defined closable anti-linear operator. Denote its closure with $\overline{S}_{\rho,\sigma}$ and its adjoint with $S_{\rho,\sigma}^*$, which is also densely defined, anti-linear and closed. Now the \textit{relative modular operator} $\Delta_{\rho,\sigma}:D(\Delta_{\rho,\sigma})\subset \cB_2(\mathcal{H})\to \cB_2({\mathcal{H}})$ \cite{art:ArakiRelModOp1977} is defined as
\begin{equation}
      \Delta_{\rho,\sigma} := S_{\rho,\sigma}^*\overline{S}_{\rho,\sigma}.
\end{equation}
Hence it can be shown to be a densely defined self-adjoint operator, which is bounded from below and satisfies for $A \in \cB_2(\HS)$
\begin{equation}
    \Delta_{\rho,\sigma}(A \sigma )=\rho A \{\sigma > 0\}\,.
    \label{equ:RelModProp}
\end{equation}  
This follows from the fact that by definition we have
\begin{align*}
    \{X\sigma+Y\{\sigma>0\}|X\in \cB(\mathcal{H}), Y\in \cB_2(\mathcal{H})\} \subset D(\Delta_{\rho,\sigma})
\end{align*} and
\begin{align*}
    \Delta_{\rho,\sigma}(A \sigma) &= S^*\overline{S}(A\sqrt{\sigma}\sqrt{\sigma})=S^*(\{\sigma>0\}(A\sqrt{\sigma})^*\sqrt{\rho})= \sqrt{\rho}(\{\sigma>0\}A^*\sqrt{\rho})^*\{\sigma>0\}\\&=\rho A \{\sigma > 0 \}.
\end{align*}
Hence, this can be seen as the analogue of the finite dimensional definition $\Delta_{\rho,\sigma}(X):=\rho X\sigma^{-1}$. For a more formal treatment and details see \cite{art:QuantumFDivergencesinVonNeumannAlgebrasI_Hiai_2018, art:ArakiRelModOp1977}.

We denote its spectral measure with $\xi^{\Delta_{\rho,\sigma}}$, i.e. $\Delta_{\rho,\sigma}=\int_{[0,\infty)}\lambda\xi^{\Delta_{\rho,\sigma}}(d\lambda)$.
Note, that if $\rho, \sigma \in \DM$, the measure
\begin{align}
    \lambda\Tr[\sigma\xi^{\Delta_{\rho,\sigma}}(d\lambda)]=\Tr[\rho\xi^{\Delta_{\rho,\sigma}}(d\lambda)]
\end{align} is a probability measure.

\subsubsection{Quantum divergences}
In the following {$\rho,\sigma\in\P(\mathcal{H})$ are positive trace-class operators on a separable Hilbert space $\mathcal{H}$.
A \textit{generalized quantum divergence} is a map $\D:\P(\HS)\times\P(\HS)\to \reals$} 
which satisfies the data-processing inequality (DPI), that is for any quantum channel $\Lambda:\P(\mathcal{H})\to \P(\mathcal{K})$ it holds that
\begin{align}
    \D(\Lambda(\rho)||\Lambda(\sigma)) \leq \D(\rho||\sigma),
\end{align} for any two states $\rho,\sigma\in\P(\mathcal{H})$.\footnote{We abuse notation a bit here, in that formally we have not defined $\D$ on $\P(\mathcal{K})$. One way to formally deal with this would be to require $\D$ to be defined on $\P(\HS)$ for all Hilbert spaces $\HS$. Equivalently, since all infinite dimensional separable Hilbert spaces are isomorphic (and hence unitarily equivalent) and we can embed any finite dimensional Hilbert space in an infinite dimensional one, we can argue that the given definition indeed covers all Hilbert spaces if we add that it should be invariant under joint unitary transformations and embeddings (this cans also easily be seen as necessary for the data-processing inequality to hold).} For any such quantum divergence, we can define its action on any two (not necessarily normalized) classical probability distributions $p, q \in \ell_1$, $p, q \geq 0$ as
\begin{equation}
    \D(p\|q) = \D\bigg(\sum_i p_i \ketbra{i}{i} \bigg\|\sum_i q_i \ketbra{i}{i}\bigg)
\end{equation}
where $\{\ket{i}\}_i$ is a countable orthonormal basis of $\HS$, and it is easy to see that the definition is independent of the choice of basis.

We will use the following divergences throughout the rest of the paper:

For $0<\epsilon\leq1$ the \textit{hypothesis testing divergence} is defined as
\begin{align}
    D_H^\epsilon(\rho\|\sigma):=-\log\left(\inf_{0\leq F\leq\1}\{\Tr[F\sigma]|\Tr[F\rho]\geq 1-\epsilon\}\right).
\end{align}
It is the negative logarithm of the minimal type II error, $\beta=\Tr[F\sigma]$, of a binary quantum hypothesis test between $\rho$ and $\sigma$ under the constraint that the type I error,  $\alpha=\Tr[(\1-F)\rho]$, is not larger than a fixed threshold value 
$\epsilon>0$.
It satisfies the DPI by a standard argument, which is repeated in appendix \ref{app:HTD} for the reader's convenience. 
The \textit{max-relative entropy} \cite{art:DattaMaxRelEnt} is defined as 
\begin{align}
    D_\text{max}(\rho||\sigma) :=\inf\{\lambda\in\mathbb{R}|\rho\leq 2^\lambda\sigma\} = \log\|\sigma^{-\frac{1}{2}}\rho\sigma^{-\frac{1}{2}}\|.
\end{align}
It satisfies the DPI and the triangle inequality \cite{art:DattaMaxRelEnt}, a proof is included in appendix \ref{app:MaxRel}. \\
For $\epsilon\in(0,1)$, $\rho\in\mathcal{D}(\mathcal{H})$ define the $\epsilon$-Sine-distance-ball around $\rho$ as
\begin{align*}
    B^\epsilon(\rho):=\{\tilde\rho\in\mathcal{D}(\mathcal{H})|P(\rho,\tilde{\rho})\leq \epsilon\}.
\end{align*}
This is a convex subset of $\mathcal{D}(\mathcal{H})$, which is (norm)-compact if and only if $\mathcal{H}$ is finite dimensional. \\
Now, the \textit{smoothed max-relative entropy} \cite{art:DattaMaxRelEnt} is defined as 
\begin{align}
    D_\text{max}^\epsilon(\rho||\sigma) := \inf_{\tilde{\rho}\in B^\epsilon(\rho)}D_\text{max}(\tilde{\rho}||\sigma).
\end{align}
Note that this is sometimes defined differently in the literature, where the smoothing is sometimes taken over the $2\epsilon$-trace-norm-ball instead of the $\epsilon$-Sine-distance-ball, and may sometimes include sub-normalized states as well. 
The smoothed max-relative entropy also satisfies the DPI which follows directly from the DPI of the max relative entropy and the DPI of the Sine-distance.\\
The \text{standard-f-divergence} \cite{art:QuantumFDivergencesinVonNeumannAlgebrasI_Hiai_2018,art:QuantumFDivergencesViaNussbaumSzkola, art:Araki1975RelativeEO} for a convex (or concave) function $f:(0,\infty)\mapsto\mathbb{R}$ is defined as
\begin{equation} \label{def:Standard-f-divergence}
    D_f(\rho||\sigma):=\int_{(0,\infty)}f(\lambda)\langle\sqrt{\sigma}|\xi^{\Delta_{\rho,\sigma}}(d\lambda)|\sqrt{\sigma}\rangle + f(0)\Tr[\sigma\{\rho=0\}]+ f'(\infty)\Tr[\rho\{\sigma=0\}],
\end{equation}
where $\{\sigma=0\}$ is the projector onto the kernel of $\sigma$,  $\xi^{\Delta_{\rho,\sigma}}$ is the spectral measure of the relative modular operator $\Delta_{\rho,\sigma}$ on the Borel sigma-algebra on $\mathbb{R}$ introduced above, $f'(\infty):=\lim_{t \to\infty}\frac{f(t)}{t}$, $f(0):=\lim_{t\to0}f(t)$, and the inner product here is the Hilbert-Schmidt inner product $\langle A|B\rangle:=\Tr[A^*B]$ for $A,B\in \cB_2(\mathcal{H})$.
Note that for convex functions for which $f(0)=0$ and for states $\rho, \sigma$ for which supp$(\rho)\subset$ supp$(\sigma)$ it follows that 
\begin{align*}
    D_f(\rho||\sigma):=\int_{(0,\infty)}f(\lambda)\langle\sqrt{\sigma}|\xi^{\Delta_{\rho,\sigma}}(d\lambda)|\sqrt{\sigma}\rangle = \int_{(0,\infty)}f(\lambda)\Tr[\sigma\xi^{\Delta_{\rho,\sigma}}(d\lambda)] & = \langle\sqrt{\sigma}|f(\Delta_{\rho,\sigma})|\sqrt{\sigma}\rangle \\ &=  \langle\sqrt{\rho}|f(\Delta_{\rho,\sigma})\Delta_{\rho,\sigma}^{-1}|\sqrt{\rho}\rangle  > -\infty.
    \end{align*}
The advantage of this expression is that it is manifestly well-defined in separable Hilbert spaces and allows for relatively simple  manipulation.\footnote{Note that the integral in \eqref{def:Standard-f-divergence} is the same as $\int_{(0,\infty)}\frac{f(\lambda)}{\lambda}\langle\sqrt{\rho}|\xi^{\Delta_{\rho,\sigma}}(d\lambda)|\sqrt{\rho}\rangle$.} \\
For the choice $f(\lambda)=\lambda\log\lambda$, the standard-f-divergence coincides with the important \textit{Umegaki relative entropy} \cite{petz_quasi-entropies_1986,art:QuantumFDivergencesinVonNeumannAlgebrasI_Hiai_2018}
\begin{equation}
   D(\rho||\sigma) := D_f(\rho||\sigma)= \begin{cases}
   \langle\sqrt{\rho}|\log\Delta_{\rho,\sigma}|\sqrt{\rho}\rangle=\Tr[\rho(\log\rho-\log\sigma)]
   ,\footnotemark & \text{if supp}\rho\subset\text{supp}\sigma \\ 
   \infty & \text{otherwise} \end{cases}
\end{equation}
which satisfies the DPI \cite{art:LindbaldOriginalDPIUmeghakiDiv1975CMaPh..40..147L}\footnotetext{More precisely, the trace is evaluated in the eigenbasis of $\rho$ which gives $D(\rho\|\sigma) = \sum_{i,j} |\langle e_i, f_j \rangle|^2 p_i (\log p_i - \log q_j)$, where $\rho = \sum_i p_i \kb{e_i}$ and $\sigma = \sum_j q_j \kb{f_j}$ being  the spectral decompositions of $\rho$ and $\sigma$.
   As detailed in \cite{Lindblad_RelativeEntropy_1973} the convexity of $x\mapsto x \log x$ implies that all terms of this sum are non-negative, which makes the expression well-defined.}\footnote{A slightly more readable proof of the DPI in infinite dimensions is given in \cite{art:DPIforUmeghakiPTPMapsM_ller_Hermes_2017}.}.  \\
Choosing $f_\alpha(\lambda)=\lambda^{\alpha}$ for $\alpha\in(0,1)\cup(1,\infty)$ we get the $\alpha$-\textit{Petz-Rényi} \cite{art:QuantumFDivergencesinVonNeumannAlgebrasI_Hiai_2018} divergence defined as
\begin{equation}
    D_\alpha(\rho||\sigma):= \begin{cases}
   \frac{1}{\alpha-1}\log D_{f_\alpha}(\rho||\sigma) = \frac{1}{\alpha-1}\log\langle\sqrt{\rho}|\Delta_{\rho,\sigma}^{\alpha-1}|\sqrt{\rho}\rangle & \text{if } \alpha\in(0,1)  \text{, or } \alpha\in(1,\infty) \\ &\text{and supp}\rho\subset\text{supp}\sigma \\ 
   \infty & \text{otherwise}. \end{cases}
\end{equation}
Note that in finite dimensions, this leads to
\begin{equation}
   D_\alpha(\rho||\sigma) = \frac{1}{\alpha-1}\log\Tr[\rho^\alpha\sigma^{1-\alpha}]\,.
\end{equation}

They can be shown to be monotonically increasing in $\alpha$ on $(0,1)\cup(1,\infty)$ \cite{art:Jen_ov__2018_1}, and satisfies
$\lim_{\alpha\nearrow 1}D_\alpha(\rho||\sigma)=D(\rho||\sigma)$ and $\lim_{\alpha\searrow 1}D_\alpha(\rho||\sigma)=D(\rho||\sigma)$ if $\exists \alpha_0>1$ such that $D_{\alpha_0}(\rho||\sigma)<\infty$ \cite{art:SandwichedRenyInInfDim_RenyDivergencesAsWeightedNonCommutative_Berta_2018}.

The \textit{geometric Renyi divergence} (also refered to as the \textit{maximal Renyi divergence})  \cite{art:MAXIMALFDIVERGENCEmatsumoto2018new} to any  $\alpha \in [0, 1) \cup (1, \infty)$ can be defined in infinite dimensions through minimal reverse tests \cite{hiai_quantum_2019}. For two positive trace-class operators $\rho, \sigma \in \P(\HS)$, we say that $(\Gamma, g, h, \mu)$ is a reverse test of $(\rho, \sigma)$, if $(X, \mu)$ is a $\sigma$-finite measure space, $g, h \in L^1(X, \mu)$ with $g, h \geq 0$, and $\Gamma$ is a positive trace-preserving map $\Gamma: L^1(X, \mu) \to \cB_1(\HS)$ such that $\Gamma(g) = \rho$, $\Gamma(h) = \sigma$. By $\Gamma$ being positive, we mean that $\Gamma(f) \geq 0$ for any $f \geq0$, and by $\Gamma$ being trace-preserving we mean that $\Tr(\Gamma(f)) = \int f d\mu$. 

We can then define the geometric \renyi trace function via the following optimization problem
\begin{equation}
    \widehat{S}_{\alpha}(\rho\|\sigma) \coloneqq \min_{\Gamma, g, h, \mu}\{ S_{\alpha}^{\mu}(g\|h)\; |\; (\Gamma, g, h, \mu) \text{ is a reverse test of } (\rho, \sigma)\}
\end{equation}
where 
\begin{equation}
    S_{\alpha}^{\mu}(g\|h) = \int h \br{g \over h}^\alpha d \mu\,,
\end{equation}
The geometric \renyi divergence is then defined as
\begin{equation}
    \widehat{D}_{\alpha}(\rho\|\sigma) \coloneqq \frac{1}{\alpha - 1} \log \widehat{S}_\alpha(\rho\|\sigma)\,.
\end{equation}

In \cite{hiai_quantum_2019}, this is studied in detail in the case where $x \mapsto x^{\alpha}$ is operator convex (i.e. $\alpha \in (1, 2]$), and for most of our applications we will restrict to this case and use many of the properties established there. It is easy to see that for any $\alpha$ the geometric \renyi divergence is the largest quantum $\alpha$-\renyi divergence that still satisfies the data-processing inequality. In the range $\alpha \in (1, 2]$, we have the following inequalities:
\begin{align}\label{eq:relation_renyi_entropies}
    D_\alpha(\rho||\sigma)\leq \widehat{D}_\alpha(\rho||\sigma)\leq D_\text{max}(\rho||\sigma).
\end{align}

\subsubsection{Quantum channel divergences}\label{sec:quantum_channel_divergences}

Given a quantum state divergence $\D$, there are different ways to extend it to a divergence on quantum channels. If $\D$ is a quantum state divergence and $\Lambda_1,\Lambda_2:\P(\mathcal{H}_A)\to\P(\mathcal{K})$ are two quantum channels, then the \textit{(stabilized) quantum channel divergence} is defined as
\begin{align}
    \D(\Lambda_1||\Lambda_2):=\sup_{\rho\in\mathcal{D}(\mathcal{H}_R\otimes\mathcal{H}_A)}\D((\id_R\otimes\Lambda_1)(\rho)||(\id_R\otimes\Lambda_2)(\rho)),
\end{align} where the supremum is in principle taken over all possible auxiliary Hilbert spaces $\mathcal{H}_R$ of arbitrary dimension and all states therein. One can show by a standard argument that the DPI of the divergence $\D$ ensures that the supremum can be restricted to pure states and $\mathcal{H}_R$ may be taken to be isomorphic to $\mathcal{H}_A$ by Schmidt decomposition.
For a given divergence $\D$ its \textit{regularized quantum channel divergence} is defined as 
\begin{align*}
    \D^{\text{reg}}(\Lambda_1||\Lambda_2):=\sup_{n\in\naturals}\frac{1}{n}\D(\Lambda_1^{\otimes n}||\Lambda_2^{\otimes n}).
\end{align*}
Note that if $\D$ is superadditive (i.e.\ $\D(\rho_1\otimes\rho_2\|\sigma_1 \otimes \sigma_2) \geq \D(\rho_1\|\sigma_1) + \D(\rho_2\|\sigma_2)$), the supremum over $n$ can be replaced with a limit $n \to \infty$.
\subsection{Quantum channel discrimination, discrimination strategies, and notation}

In binary hypothesis testing, if $H_0$ and $H_1$ denote the {\em{null hypothesis}} and the {\em{alternate hypothesis}}, respectively, then the type I error ($\alpha$) is the probability of accepting $H_1$ when $H_0$ is true, while the type II error ($\beta$) is the probability of accepting $H_0$ when $H_1$ is true.

Quantum channel discrimination is the binary hypothesis testing task in which $H_0:\E$ and $H_1:\F$, where $\E$ and $\F$ are two quantum channels.
In the asymmetric setting, the aim is to minimise the type II error under the condition that the type I error is below some prescribed 
threshold value $\epsilon>0$. 

There are two main ways in which we can set up our decision experiment when allowed $n$ uses of the black-box channel. \\
\\
\textit{Parallel} (or \textit{non-adaptive}) quantum channel discrimination, for a given pair of channels $\mathcal{E}$ and $\mathcal{F}$, allows only parallel uses of the given (unknown) channel. It effectively amounts to discriminating between a single use of $\mathcal{E}^{\otimes n}$ or $\mathcal{F}^{\otimes n}$ (for some $n \in \mathbb{N}$), using a single joint input state over all the $n$ parallel uses of the channel and some ancilla system. Generically this input state is entangled between the $n$ channel inputs and the ancilla system. It follows from the definition of the hypothesis testing divergence that the decay rate of the optimal type II error per channel use, when using such a parallel setup, is given by
\begin{align}
    \frac{1}{n}D_H^\epsilon(\mathcal{E}^{\otimes n}||\mathcal{F}^{\otimes n})\,.
\end{align}
By the infinite dimensional version of Quantum Stein's lemma for channels (\autoref{thm:steins} below) we show that in the limits $n\to\infty$ and $\epsilon\to0$, this is equivalent to the regularized channel divergence of these two channels. \\
\textit{Adaptive} (or \textit{sequential}) quantum channel discrimination strategies are more general. They use the $n$ black-box channels (all identical but unknown) sequentially, and allow for the input state for any use of the channel to be chosen based on the outputs of previous channel uses. Specifically, we can perform arbitrary manipulations on the previous output state and some additional ancilla system, described by preparation channels, to determine the next input state. For an adaptive strategy that discriminates between $\mathcal{E}_{A\to B}$ and $\mathcal{F}_{A\to B}$ using $n$ queries, we denote the set of these preparation channels as $\{\Lambda_{i;B_iR_i\to A_{i+1}R_{i+1}}\}_{i=2}^n$, where the $\{R_i\}_{i=1}^n$ are some ancillary quantum systems of (possibly) unbounded dimension. We call the initial input state $\rho_1=\sigma_1\in\mathcal{D}(\mathcal{H}_{A_1R})$. 
Given the channel $\mathcal{E}$, the output state after an $n$ round adaptive strategy is given by $ \E(\rho_n) = (\id_{R_n}\otimes\mathcal{E}_{A_n \to B_n })(\rho_{n; A_n R_n}) := \mathcal{E}\circ \Lambda_n\circ\mathcal{E}\circ ... \circ \Lambda_2\circ\mathcal{E}(\rho_1)$, and similarly for $\mathcal{F}(\sigma_n)$ if the channel was $\mathcal{F}$.
Here we define
\begin{align}
   &\rho_i:=\Lambda_i(\mathcal{E}(\rho_{i-1})), \\
   &\sigma_i:=\Lambda_i(\mathcal{F}(\sigma_{i-1})), 
\end{align} for $i\in\{2,...,n\}$.
\begin{figure}
    \centering
    \includegraphics{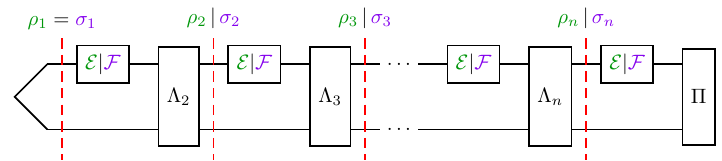}
    \caption[width=\linewidth]{Illustration of a general adaptive protocol with $n$ uses of a black-box channel. The top row makes use of the given black-box $\E|\F$, which is either $\E$ or $\F$, while the bottom row depicts the ancilla system $R$. At various stages in the protocol, the green states $\rho_i$ occur if the channel is $\E$ and the purple states $\sigma_i$ occur if the channel is $\F$; here $i \in {1, 2, \ldots n}$. The intermediate \textit{preparation channels} are denoted by $\Lambda_i$ and their action results in the preparation of the input state for the $i^{th}$ instance of the black-box channel depending on the output of the $(i-1)^{th}$ instance. This input state could be entangled with the ancilla system $R$. The final state $\E(\rho_n)$ or $\F(\sigma_n)$ is then subjected to a binary measurement described by the POVM $\Pi$.}
    \label{fig:plot_adaptive}
\end{figure}
For such an adaptive strategy the decay rate in 
$n$ of the optimal type II error of this discrimination task is then given by
\begin{align}
   \sup_{\rho_1,\{\Lambda_i\}_i} \frac{1}{n}D_H^\epsilon(\mathcal{E}(\rho_n)||\mathcal{F}(\sigma_n))
\end{align} where the supremum is over all possible initial states $\rho_1=\sigma_1$ and all preparation channels $\{\Lambda_i\}_{i=2}^n$.

\subsection{Previous finite dimensional results}
To facilitate comparison with the infinite dimensional results we state in the next section, we include here a list of what has previously been proven for quantum channel discrimination in finite dimensions.

The following Stein's lemma established the optimal asymmetric error exponent (i.e. decay rate of the type II error in the number of channel uses) for adaptive and parallel channel discrimination strategies.
\begin{theorem}[Stein`s lemma for quantum channels, \cite{wang_resource_2019, wilde_amortized_2020}]
	Let $\HS_A$ and $\HS_B$ be two finite dimensional Hilbert spaces, and $\E$, $\F : \P(\HS_A) \to \P(\HS_B)$ be two quantum channels.
	The asymmetric asymptotic error exponent of discirminating these two channels with a parallel strategy is then given by
	\begin{equation}
			\lim_{\varepsilon \to 0}\lim_{n \to \infty} {1 \over n} D_H^{\varepsilon}(\E\n\|\F\n) = D^{\mathrm{reg}}(\E\|\F)
	\end{equation}
	Moreover, the asymmetric asymptotic error exponent of discriminating these two channels with an adaptive strategy is given by
	\begin{equation}
		D^A(\E\|\F) = \sup_{\rho_{RA}, \sigma_{RA}}\big( D(\E(\rho)\|\F(\sigma)) - D(\rho\|\sigma)\big)
	\end{equation}
	where the supremum goes over all possible sizes of reference systems $R$. 
\end{theorem}

The equivalence of adaptive and parallel error exponents was then shown using the following chain rule for the quantum relative entropy. 
\begin{theorem}[Chain Rule, \cite{fang_chain_2020}]
		Let $\HS_A$ and $\HS_B$ be two finite dimensional Hilbert spaces, and $\E$, $\F : \P(\HS_A) \to \P(\HS_B)$ be two quantum channels.
		Then, for all states $\rho, \sigma \in \DM[\HS_A]$
		\begin{equation}
			D(\E(\rho)\|\F(\sigma)) \leq D(\rho\|\sigma) + D^{\mathrm{reg}}(\E\|\F)
		\end{equation}
		which implies that $D^{\mathrm{reg}} = D^A$.
\end{theorem}

While this shows the equivalence of the asymptotic rates for adaptive and parallel strategies, it does not allow any conclusions about the relation of adaptive and parallel strategies for finite number of channel uses. Such a relation for finite number of channel uses is given by the following \enquote{one-shot} result:
\begin{theorem}[{\cite[Corollary 4]{art:Bjarne}}]\label{thm:bjarne_corr}
	Let $\HS_A$ and $\HS_B$ be two finite dimensional Hilbert spaces, and let $\E, \F: \P(\HS_A)  \to \P(\HS_B)$ be two quantum channels such that $D_{\max}(\E\|\F) <  \infty$. Let there be an adaptive discrimination protocol with $n$ channel uses,  specified by the final states $\E(\rho_n)$ and $\F(\sigma_n)$ (see above for a more thorough explanation of this notation). Then, for any $\alpha_a$ , $\alpha_p \in (0,1)$ there exists a parallel input state $\nu_m \in \DM[\HS_A^{\otimes m} \otimes \HS_{A'}^{\otimes m}]$ such that
    \begin{equation}
        	    \frac{1}{m} D_H^{\alpha_p}(\E^{\otimes m}(\nu_m)\|\F^{\otimes m}(\nu_m))  \geq \frac{1 - \alpha_a}{n} D_H^{\alpha_a}(\E(\rho_n)\|\F(\sigma_n)) - \frac{C n}{\sqrt{m}}\log\br{8\over \alpha_p} - \frac{1}{n}.
    \end{equation}
    That is, the type~II error rate of the parallel protocol is essentially at least as good as the adaptive one modulo an additional error term, which decays as $m \to \infty$. 
    We have
    \begin{equation}
        C \coloneqq 7 \log(2^{D_2(\E\|\F)} + 2) = 7 \log(2^{\widehat{D}_2(\E\|\F)} + 2)\leq 7 \log(2^{D_{\max}(\E\|\F)} + 2)\,.
    \end{equation}
    
\end{theorem}
As the proof of an infinite dimensional version of this one-shot result constitutes a large part of this paper, let us briefly discuss the technical steps required for the proof, and how we generalized them to infinite dimensions.
For the derivation of this result in \cite{art:Bjarne} the following inequalities \eqref{equ:ImpIneq1}-\eqref{equ:ImpIneq5} were instrumental. 

\begin{equation}
	D_\text{max}^\epsilon(\rho||\sigma)\leq D_H^{1-\epsilon^2}(\rho||\sigma)+\log\left(\frac{1}{1-\epsilon^2}\right)
	\label{equ:ImpIneq1}
\end{equation}
For $\alpha\in(1,\infty)$ it holds that
\begin{equation}
	D_\text{max}^\epsilon(\rho||\sigma)\leq D_\alpha(\rho||\sigma)+\frac{2}{\alpha-1}\log\left(\frac{1}{\epsilon}\right)+\log\left(\frac{1}{1-\epsilon^2}\right).
	\label{equ:ImpIneq2}
\end{equation}
For $\alpha\in(0,1)$ it holds that
\begin{equation}
	D_\text{max}^\epsilon(\rho||\sigma) \geq D_\alpha(\rho||\sigma) + \frac{2}{\alpha-1}\log\left(\frac{1}{1-\epsilon}\right).
	\label{equ:ImpIneq3}
\end{equation}

\begin{equation}
	D_H^\epsilon(\rho||\sigma)\leq \frac{1}{1-\epsilon}[D(\rho||\sigma)+h(\epsilon)]
	\label{equ:ImpIneq4}
\end{equation}
where $h(\epsilon)=-\epsilon\log\epsilon-(1-\epsilon)\log(1-\epsilon)$ is the binary entropy.

For any two quantum channels $\mathcal{E},\mathcal{F}$ on finite dimensional Hilbert spaces and for any $\epsilon,\epsilon^\prime$, we can find a state $\nu$, s.t.
\begin{equation}
	D_\text{max}^{\epsilon+\epsilon^\prime}(\mathcal{E}(\rho)||\mathcal{F}(\sigma)) \leq D_\text{max}^{\epsilon^\prime}(\mathcal{E}(\nu)||\mathcal{F}(\nu)) + D_\text{max}^{\epsilon}(\rho||\sigma).
	\label{equ:ImpIneq5}
\end{equation}

The original proofs of equations \eqref{equ:ImpIneq1} and \eqref{equ:ImpIneq2}, given in \cite{art:MinimaxMaxDivergenceAnshu_2019,wang_resource_states_2019}, respectively, are based on a minimax expression for the smoothed max-relative entropy \cite{art:MinimaxMaxDivergenceAnshu_2019} which requires Sion's minimax theorem. However, the generalisation of this proof technique to infinite dimensional Hilbert spaces is not straightforward. Hence, we give a new proof of these inequalities in Corollary \ref{COR:1} and \ref{COR:2} respectively. Even though the original proof of \eqref{equ:ImpIneq3} in \cite{wang_resource_states_2019} was not done explicitly in infinite dimensions, the same proof works for the latter case. Its proof (in infinite dimensions) is spelled out in Lemma~\ref{lemma:SmoothedDmaxLowerBoundWW}. The proof of inequality \eqref{equ:ImpIneq4} is a convenient application of the DPI for the Umegaki relative entropy (as shown in \cite{art:LindbaldOriginalDPIUmeghakiDiv1975CMaPh..40..147L}), and thus also directly holds in infinite dimensions. A proof may be found in \cite{wang_one-shot_2012}. Equation \eqref{equ:ImpIneq5} (proved in \cite[Lemma 8]{art:Bjarne}) will have to be slightly modified in the infinite dimensional setting, via the same proof idea as in \cite{art:Bjarne}; its extension can be found in Lemma \ref{lemma:BjarneLemma8}.

\section{Main results and infinite dimensional entropic inequalities}
Throughout this section, we will assume that $\HS_A$, $\HS_B$ are two separable Hilbert spaces and $\E$, $\F : \P(\HS_A) \to \P(\HS_B)$ are two CPTP maps such that $\Dg(\E\|\F) < \infty$ for some $\alpha > 1$, where we will discuss this last condition in detail towards the end of this section. We then get the following results regarding the discrimination of these two quantum channels in infinite dimensions.

\begin{theorem}[Stein's lemma for infinite dimensional channels]\label{thm:steins}
	Let $\E$, $\F: \P(\HS_A) \to \P(\HS_B)$ be two quantum channels between separable Hilbert spaces $\HS_A$, $\HS_B$. Then, the asymmetric asymptotic error exponent of discriminating these two channels with a parallel strategy is given by
	\begin{equation}
		\lim_{\varepsilon \to 0}\lim_{n \to \infty} {1 \over n} D_H^{\varepsilon}(\E\n\|\F\n) = D^{\mathrm{reg}}(\E\|\F)\,.
	\end{equation}
	Moreover, if $\Dg(\E\|\F) < \infty$ for some $\alpha > 1$, then this is also equal to the asymptotic error exponent using adaptive strategies and hence adaptive strategies offer no asymptotic advantage.
\end{theorem}

This establishes the role of the regularised channel relative entropy as the asymptotic asymmetric discrimination exponent also in infinite dimensions. Additionally to this asymptotic result, we have the following relation between adaptive and parallel exponents for finite number of channel uses:

\begin{theorem}[One-shot conversion of adaptive to parallel strategies]
\label{thm:main}
Let $\mathcal{E},\mathcal{F}:\P(\mathcal{H}_A)\to\P(\mathcal{H}_B)$ be two quantum channels between separable Hilbert spaces $\mathcal{H}_A$ and $\mathcal{H}_B$, s.t.  $\widehat{D}_{\alpha}(\mathcal{E}\|\mathcal{F})<\infty$ for some $\alpha>1$. Then for any given adaptive discrimination strategy specified by the final states $\E(\rho_n)$ and $\F(\sigma_n)$ (see above for a more thorough explanation of this notation) and for any $\alpha_a, \alpha_p\in (0,1)$ and any $\mu > 0$ there exists a parallel input state $\nu \in \DM[\HS_A^{\otimes n} \otimes \HS_{A'}^{\otimes n}]$ such that

\begin{align}
    \frac{1}{n}D^{\alpha_a}_H(\mathcal{E}(\rho_n)\|\mathcal{F}(\sigma_n))\leq \frac{1}{1-\alpha_a}\frac{1}{m}D^{\alpha_p}_H(\mathcal{E}^{\otimes m}(\nu)\|\mathcal{F}^{\otimes m}(\nu))+f(\alpha_a,\alpha_p,m,n,\mu),
\end{align}
where the error term is
\begin{align}
    f(\alpha_a,\alpha_p,m,n,\mu)=\frac{1}{1-\alpha_a}\left[\frac{Cn}{\sqrt{m}}\log\left(\frac{8}{\alpha_p}\right)+\frac{1}{m}\left[\log\frac{1}{\alpha_p}-\log\left(1-\frac{\alpha_p}{4}\right)+\mu\right]+\frac{h(\alpha_a)}{n}\right]
\end{align}
and
\begin{equation}
    C \leq 8 \inf_{\alpha \in (1,2]} \frac{1}{\alpha - 1} \log\left(2^{(\alpha - 1) \widehat{D}_{\alpha}(\E\|\F)} + 2\right)\,.
\end{equation}

This means that the type II error scaling of the parallel strategy is at least as good as the one of the adaptive strategy up to an error term vanishing for $m,n\to\infty$.

\end{theorem}

\begin{remark}
Note that \cite[Theorem 6]{art:Bjarne} also proved slightly tighter, but much more complicated version of \autoref{thm:bjarne_corr}. A similar generalization to infinite dimensions along the lines of our proof of \autoref{thm:main} also holds for this more complicated version. 
\end{remark}

We will use \autoref{thm:main} to prove the asymptotic equivalence of adaptive and parallel strategies in \autoref{thm:steins}, which means that unlike the original argument in finite dimensions \cite{fang_chain_2020, wang_resource_2019}, we do not go the way of establishing an amortized expression for the adaptive exponent and then applying a chain rule. However, we still consider the chain rule for the quantum relative entropy to be of independent interest, and we are able to establish it using our techniques under the following fairly weak condition:
\begin{theorem} \label{thm:InfDimChainRule}
	Let $\mathcal{H},\mathcal{K}$ be separable Hilbert spaces, and $\E, \F : \P(\HS) \to \P(\mathcal{K})$ be two quantum channels, such that $\Dg(\E\|\F) < \infty$ for some $\alpha > 1$. Then,  for any two states $\rho,\sigma\in\mathcal{D}(\mathcal{H})$ it holds that:
	\begin{equation}
		D(\E(\rho)\|\F(\sigma)) \leq D(\rho\|\sigma)+ D^{\mathrm{reg}}(\E\|\F).
	\end{equation}
\end{theorem}
\begin{remark}
	In the updated version of \cite{fawzi_asymptotic_2023} the authors prove this chain rule under the stronger condition $D_{\max}(\E\|\F) < \infty$, and hence \autoref{thm:InfDimChainRule} can be seen as an improvement of their result. Please see our remarks at the end of this section for a discussion on why this less restrictive finiteness condition can be relevant.
\end{remark}

A key technical ingredient of our argument, and a major contribution of this paper is the following chain rule for the geometric \renyi divergence, which we prove in \autoref{sec:geometric}:
\restate{geometricChainRule}

\subsection{Finiteness condition}\label{sec:finiteness_cond}
As mentioned previously, we are able to show the equivalence of adaptive and parallel discrimination strategies in infinite dimensions only under the condition that $\Dg(\E\|\F) < \infty$ for some $\alpha > 1$, whereas in finite dimensions no such condition was needed.
The reason for this is that in finite dimensions, most relative entropies are infinite at the same time, and only if a support-condition is violated. Specifically,
\begin{equation} \label{eq:finite_dim_divergence_infinite}
	D(\rho\|\sigma) = \infty \Leftrightarrow \mathbb{D}_{\alpha}(\rho\|\sigma) = \infty \; \forall \alpha > 1 \Leftrightarrow D_{\max}(\rho\|\sigma) = \infty \Leftrightarrow \rho \not\ll \sigma
\end{equation}
where $\mathbb{D}_{\alpha}$ is any quantum $\alpha$-\renyi relative entropy (i.e. any function that satisfies the data-processing inequality (DPI) and reduces to the classical $\alpha$-\renyi entropy on commuting states).  Hence, if our finiteness condition $\Dg(\E\|\F) < \infty$ is violated in finite dimensions, also $D(\E\|\F) = \infty$, which then implies that a parallel strategy (even a product strategy that just repeats the same input state) can achieve an asymmetric error exponent of infinity, and hence there cannot be any asymptotic adaptive advantage. 

In infinite dimensions, however, while the divergences in \eqref{eq:finite_dim_divergence_infinite} will also all be infinite if the support condition is violated, this is not the only possibility, and additionally infinity will occur if \enquote{one of the two states decays faster than the other}. 
This appears already classically. As an example, consider the following three (unnormalized) classical probability distributions on $\naturals$: $p=\{p_n\}, q=\{q_n\}$ and $r=\{p_n\}$, with
\begin{equation}
	p_n = \frac{2^{-n}}{n^2} \qquad q_n = \frac{p_n}{n} \qquad r_n = p_n 2^{-2^n}
\end{equation}
We could easily make these three probability distributions normalized by redefining them with a normalization factor, however this does not change the finiteness of the divergences, so we omit the normalization factors to keep the following calculations simpler. We have $p_n, q_n, r_n > 0$ for all $n$, and hence they satisfy the support condition. On the other hand, it is easy to see that $p_n/q_n = n$ is unbounded, and hence $D_{\max}(p\|q) = \infty$, whereas 
\begin{equation}
	D_{\alpha}(p\|q) = \frac{1}{\alpha - 1} \log\sum_n q_n \br{p_n \over q_n}^{\alpha} = \frac{1}{\alpha - 1} \log \sum_n 2^{-n} n^{\alpha - 3}
\end{equation}
is finite for all $\alpha > 1$. Similarly, 
\begin{equation}
	D_{\alpha}(p\|r) = \frac{1}{\alpha - 1} \log\sum_n 2^{(\alpha - 1) 2^{n} - n} n^{- 2}
\end{equation}
is clearly infinite for all $\alpha > 1$, whereas 
\begin{equation}
	D(p\|r) = \sum_n p_n \log\br{p_n \over r_n} = \sum_n {1 \over n^2}(1 - 2^{-n})
\end{equation}
is finite. Roughly speaking, in infinite dimensions the finiteness of these divergences is governed by how the tails of the distributions $p$ and $q$ are related to each other. Finiteness of $D_{\max}$ requires $q_n$ to be at most a constant factor smaller than $p_n$. Assuming, that $p_n$ is such that also the sequence $p'_n = p_n n^{\beta}$ is summable for some $\beta > 0$, then $D_{\alpha}(p\|q)$ will be finite for some $\alpha > 1$ if $q_n$ is at most a polynomial in $n$ smaller than $p_n$ (i.e. $p_n/q_n = \mathcal{O}(n^{\beta/\alpha})$), and $D(p\|q)$ will be finite if $q_n$ is at most exponentially smaller than $p_n$, i.e. $p_n/q_n = \mathcal{O}(2^{n^{\beta}})$. 

\subsubsection{Different conditions for quantum channel discrimination}
Given that in infinite dimensions finiteness conditions involving different relative entropies are no longer equivalent, we would like to demonstrate, in this section and the next, that there are examples of quantum channels (which are interesting in the context of channel discrimination) that satisfy some but not all of the finiteness conditions. Specifically, given two quantum channels $\E$ and $\F$ we will be looking at the the following three finiteness conditions 
\begin{enumerate}[label=(\alph*)]
    \item $D_{\max}(\E\|\F) < \infty$
    \item $\Dg(\E\|\F) < \infty$ for some $\alpha>1$
    \item $D^{\text{reg}}(\E\|\F) < \infty$
\end{enumerate}

In the context of comparing adaptive and parallel strategies, the third condition would be the desired one, as it amounts to imposing no condition at all. This follows from the fact that $D^{\mathrm{reg}} = \infty$ implies that both adaptive and parallel strategies are able to achieve a rate of infinity, and so one then obtains asymptotic equivalence of adaptive and parallel strategies for all channels. 

By choosing replacer channels (i.e. channels that output a specific state independent of the input) outputting classical states corresponding to the probability distributions $p, q$ and $r$ defined above, we see that there exist classical channels for which these three finiteness conditions are not equivalent. However, these channels are not particularly interesting from the perspective of channel discrimination and for studying the relation between adaptive and parallel strategies. This is because for replacer channels, the input state and hence also the chosen strategy is irrelevant already on the one-shot level. Additionally, for classical channels also on continuous systems (which form the classical analogue of an infite dimensional quantum systems) it is known that adaptivity does not give any asymptotic advantage, without requiring any finiteness conditions \cite{hayashi_discrimination_2009}.

Hence, in the next subsection, we give examples of fully quantum channels that are not replacer channels for which the three finiteness conditions are not equivalent. This illustrates first, that there are interesting channels satisfying only condition (b) but not (a), and hence being able to establish the equivalence of adaptive and parallel strategies under condition (b) is significant. Additionally, we find channels for which only condition (c) holds and for which we are currently unable to show the equivalence. 

We 
conjecture that, in fact, 
no condition 
is necessary also in infinite dimensions, and the asymptotic equivalence of adaptive and parallel strategies does hold for all channels also in infinite dimensions. However, we currently do not have the necessary tools to prove this conjecture. We would also like to highlight that we believe the condition (a) to be very restrictive in practice. Specifically, with this condition the data-processing inequality implies that in \emph{all bases} the diagonal elements of $\E(\nu)$ and $\F(\nu)$ have to decay while differing at most by a constant factor, and this too for all input states $\nu$. For channels which do not just output states on some finite-dimensional subspace we expect this to be the case only for very specific examples.
 
\subsubsection{Fully quantum examples}

For any state $\tau \in \DM[\HS_A]$ and any $\lambda \in [0,1]$, define the generalized depolarizing channel $\Lambda^\lambda_\tau: \P(\HS_A) \to \P(\HS_A)$ as
\begin{equation}
    \Lambda^\lambda_\tau(\omega) \coloneqq (1-\lambda) \omega + \lambda \Tr(\omega) \tau \,.
\end{equation}

\begin{lemma}\label{lem:generalized_depolarizing_upper_bound}
    For any two states $\rho, \sigma \in \DM[\HS_A]$, any $\lambda \in [0,1]$ and any quantum divergence $\D$ (i.e. any function of two positive trace-class operators that satisfies the data-processing inequality) and any $n \in \naturals$, it holds that
    \begin{equation}
        \D((\Lambda^\lambda_\rho)^{\otimes n}\|(\Lambda^\lambda_\sigma)^{\otimes n}) \leq \D(\rho^{\otimes n}\|\sigma^{\otimes n})
    \end{equation}
    where the channel divergence on the left-hand side is defined as in \autoref{sec:quantum_channel_divergences}.
\end{lemma}
\begin{proof}
    For the sake of illustration, consider first the case $n = 1$. For any normalized density matrix $\omega = \omega_{RA}$ we then have
    \begin{align}
        \D(\rho\|\sigma) &= \D(\omega_R \otimes \rho\|\omega_R \otimes \sigma) \\
        &\geq \D(\Lambda^{1-\lambda}_{\omega_{RA}}(\omega_R \otimes \rho)\|\Lambda^{1-\lambda}_{\omega_{RA}}(\omega_R \otimes \sigma)) \\
        &= \D(\lambda (\omega_R \otimes \rho)  + (1-\lambda) \omega_{RA}\|\lambda (\omega_R \otimes \sigma)  + (1-\lambda) \omega_{RA}) \\
        &= \D((\id_R \otimes \Lambda^\lambda_\rho)(\omega_{RA})\|(\id_R \otimes \Lambda^\lambda_\sigma)(\omega_{RA})),
    \end{align}
    where the first line also follows from the data-processing inequality, by using channels that add or trace out the additional state $\omega_{R}$. 
    As this holds for all states $\omega_{RA}$, it then also holds for the supremum over all states.
    
    For $n > 1$, the argument is essentially the same. For a subset $S \subset \{1, 2, ..., n\}$ and a state $\tau_{A^n} = \tau_{A_1 ... A_n}$ we write $\tau_{A_S}$ for $\Tr_{A_{S^c}}(\tau_{A^n})$, where we trace out all the systems $A_i$ whose index is not in $S$, and $S^c$ is the complement of $S$. Given any state $\omega_{R^nA^n}$, we can then define a channel $\mathcal{M}: A^n \to R^n A^n$ via
    \begin{equation}
        \mathcal{M}(\tau_{A^n}) \coloneqq \sum_{S \subset \{1, ..., n\}} \lambda^{|S|}(1 - \lambda)^{|S^c|} \tau_{A_S} \otimes \omega_{R_S} \otimes \omega_{R_{S^c}A_{S^c}}\,,
    \end{equation}
    with the idea being that 
    \begin{align}
        (\id_R \otimes \Lambda^\lambda_{\rho})^{\otimes n}(\omega_{R^n A^n}) &= \mathcal{M}(\rho^{\otimes n}) \\
        (\id_R \otimes \Lambda^\lambda_{\sigma})^{\otimes n}(\omega_{R^n A^n}) &= \mathcal{M}(\sigma^{\otimes n}) 
    \end{align}
    and hence the claim follows again from an application of the data-processing inequality:
    \begin{equation}
        \D(\rho^{\otimes n}\|\sigma^{\otimes n}) \geq \D(\mathcal{M}(\rho^{\otimes n})\|\mathcal{M}(\sigma^{\otimes n})) = \D((\id_R \otimes \Lambda^\lambda_{\rho})^{\otimes n}(\omega_{R^n A^n})\|(\id_R \otimes \Lambda^\lambda_{\sigma})^{\otimes n}(\omega_{R^n A^n}))\,.
    \end{equation}
\end{proof}
Note specifically that for any divergence $\D$ that is additive on states (e.g. $\D = D$), this implies 
\begin{equation}
    \D^{\mathrm{reg}}(\Lambda^\lambda_\rho\|\Lambda^\lambda_\sigma) \leq \D(\rho\|\sigma)\,.
\end{equation}

To continue constructing our examples, given any orthonormal basis $\{\ket{a_i}\}_{i = 0}^{\infty}$ of a Hilbert space, we construct the basis $\{\ket{b_i}\}_{i = 0}^{\infty}$ as
\begin{equation}
    \ket{b_i} \coloneqq \begin{cases}{1 \over \sqrt{2}}(\ket{a_i} + \ket{a_{i + 1}}) & i \text{ even } \\ {1 \over \sqrt{2}}(\ket{a_{i - 1}} - \ket{a_i}) & i \text{ odd\,. }\end{cases}
\end{equation}

For any classical probability distribution $p \in \ell_1$ we then define the states
\begin{equation}
    \rho_p \coloneqq \sum_i p_i \ketbra{a_i}{a_i} \qquad \sigma_p \coloneqq \sum_i p_i \ketbra{b_i}{b_i}\,.
\end{equation}

\begin{lemma}
Let $\D$ be one of $D, \Dg$ (with $\alpha \in (1,2]$) or $D_{\max}$, and $p, q \in \ell_1$ be normalized probability distributions. Then, for all $\lambda \in (0, 1]$
\begin{equation}
     \D(\rho_p\|\sigma_q) = \infty \Rightarrow \D(\Lambda^\lambda_{\rho_p}\|\Lambda^\lambda_{\sigma_q}) = \infty\,.
\end{equation}
\end{lemma}
\begin{proof}
    Throughout this proof we will write $\rho \coloneqq \rho_p$ and $\sigma \coloneqq \sigma_q$ for simplicity. Let us start with just chosing the reference system of the input state as trivial and picking the input state $\omega = \sigma$. Then
    \begin{equation}
        \D(\Lambda^\lambda_{\rho}(\omega)\|\Lambda^\lambda_{\sigma}(\omega)) = \D((1 - \lambda)\sigma + \lambda \rho\|\sigma)
    \end{equation}
    If $\D = D_{\max}$ we can use the monotonicity of $D_{\max}$ in the first variable to conclude that 
    \begin{equation}
        \D((1 - \lambda)\sigma + \lambda \rho\|\sigma) \geq \D(\lambda \rho \| \sigma) = \infty\,.
    \end{equation}
    If $\D = D$, we can use the almost-concavity of $D$ in the first argument (we provide a proof in \autoref{lem:entropy_almost_concave}) to find
    \begin{equation}
        \D((1 - \lambda)\sigma + \lambda \rho\|\sigma) \geq (1 - \lambda)\D(\sigma\|\sigma) + \lambda \D(\rho\|\sigma) - h(\lambda) = \infty
    \end{equation}
    where $h$ is the binary entropy. 
    \newcommand{\M}{\mathcal{M}}
Note that so far we have not used any of the structure of the states $\rho$ and $\sigma$. \\
For $\D = \Dg$ we could use a similar monotonicity argument for $\alpha \in [0,1]$, but for $\alpha > 1$, the function $t \mapsto t^{\alpha}$ is unfortunately not operator monotone, so we need a different argument, which is where the assumptions on the states come in. Let $j$ be such that $q_j > 0$ and then pick the state $\omega = \ketbra{b_j}{b_j}$. Let also $k_1, k_2$ be the odd and even index of the block associated to $j$, i.e. $k_1 = 2\left\lfloor {j \over 2}\right\rfloor$, $k_2 =2\left\lfloor {j \over 2}\right\rfloor + 1$. Consider the subspace of our Hilbert space spanned by $\{\ket{a_i}\}_{i=k_1, k_2}$, and its orthogonal complement (this is the same as the span of $\{\ket{b_i}\}_{i=k_1, k_2}$ and its orthogonal complement), and let $\M$ be the POVM measurement channel between these two subspaces, i.e. $\M(\nu) = \sum_{i = 1,2} \Pi_i\nu\Pi_i$, where the $\Pi_1$ projects onto the subspace, and $\Pi_2$ projects onto its orthogonal complement. Then, 
    \begin{align}
        \D(\Lambda^\lambda_\rho(\omega)\|\Lambda^\lambda_\sigma(\omega)) &\geq \D(\M \circ\Lambda^\lambda_\rho(\omega)\|\M \circ \Lambda^\lambda_\sigma(\omega)) \\ &\geq \D(\Pi_2\Lambda^\lambda_\rho(\omega)\Pi_2\|\Pi_2\Lambda^\lambda_\sigma(\omega)\Pi_2)\\ &= \D(\lambda \Pi_2\rho\Pi_2\|\lambda\Pi_2 \sigma\Pi_2),\label{eq:example_geom_infinite_last_line}
    \end{align}
    where we used the data-processing inequality,  \autoref{lem:generalized_direct_sum_decomp}, and the fact that $\omega$ lies in the kernel of $\Pi_2$. Since $\D(\rho\|\sigma)$ is infinite and $\rho$ and $\sigma$ are block-diagonal in the decomposition, either the term in \eqref{eq:example_geom_infinite_last_line} has to be infinite (which proves our statement), or $\D(\Pi_1\rho\Pi_1\|\Pi_1\sigma\Pi_1)$ is infinite, which (since the subspace is finite dimensional) implies that $\Pi_1\rho\Pi_1 \not\ll \Pi_1\sigma\Pi_1$, which implies $\rho \not \ll \sigma$. Since $q_j > 0$, we have $\omega \ll \sigma$, and hence this also implies $\Lambda^\lambda_\rho(\omega) \not \ll \Lambda^\lambda_\sigma(\omega)$ and so also the channel divergence has to be infinite.
\end{proof}

\begin{lemma}
Let $\D$ be one of $D, \Dg$ (with $\alpha \in (1,2]$) or $D_{\max}$. Let $p, q \in \ell_1$ be positive, and define the following variants of q, which take the minimum/maximum over a block of two indices:
\begin{equation}
    q^{\uparrow}_i \coloneqq \max\{q_{2\lfloor{i \over 2}\rfloor}, q_{2\lfloor{i \over 2}\rfloor +1}\} \qquad
    q^{\downarrow}_i \coloneqq \min\{q_{2\lfloor{i \over 2}\rfloor}, q_{2\lfloor{i \over 2}\rfloor +1}\}\,.
\end{equation}
Then,
\begin{equation}\label{eq:example_up_downarrow}
\D(p\|q^{\uparrow}) \leq \D(\rho_p\|\sigma_q) \leq \D(p\|q^{\downarrow})\,.
\end{equation}
\end{lemma}

\begin{proof}
    All the three divergences we consider are anti-monotonous in the second variable. For $D_{\max}$ this is obvious from the definition, for $D$ this is shown in \cite[Theorem 4.1]{hiai_quantum_2019}, and for $\Dg$ we show it in \autoref{lem:geometric_anti_monotonicity}. It is easy to see that $\sigma_{q^\downarrow} \leq \sigma_q \leq \sigma_{q^\uparrow}$, but $\sigma_{q^\downarrow} = \rho_{q^\downarrow}$, and $\sigma_{q^\uparrow} = \rho_{q^\uparrow}$, which implies the desired statement. 
\end{proof}

The probability distributions $p, q, r$ defined in the previous section are such that the three terms in \eqref{eq:example_up_downarrow} are all either finite or infinite at the same time. Hence, all these lemmas together imply that for $\D$ one of $D, \Dg$ (with $\alpha \in (1,2]$) or $D_{\max}$, $\lambda \in (0,1]$ and $p, q$ one of these three probability distributions it holds that
\begin{equation}
    \D(\Lambda^\lambda_{\rho_p}\|\Lambda^\lambda_{\sigma_q}) = \infty \Leftrightarrow  \D(p\|q) = \infty\,.
\end{equation}

Additionally, since $D \leq D^{\mathrm{reg}}$ and the upper bound in \autoref{lem:generalized_depolarizing_upper_bound} also includes the statement for tensor products of the channels, we get
\begin{equation}
    D^{\mathrm{reg}}(\Lambda^\lambda_{\rho_p}\|\Lambda^\lambda_{\sigma_q}) = \infty \Leftrightarrow  D(p\|q) = \infty\,.
\end{equation}

Specifically, for $\lambda \in (0,1]$ and
\begin{equation}
    p_n = {2^{-n}n^{-2} \over \sum_{k} 2^{-k}k^{-2}} \qquad q_n = {2^{-n}n^{-3} \over \sum_{k} 2^{-k}k^{-3}}
\end{equation}
we have $\Dg(\Lambda^\lambda_{\rho_p}\|\Lambda^\lambda_{\sigma_q}) < \infty$ for all $\alpha \in (1, 2]$, while $D_{\max}(\Lambda^\lambda_{\rho_p}\|\Lambda^\lambda_{\sigma_q}) = \infty$. 
Similarly, for 
\begin{equation}
    r_n = {2^{-n}2^{-2^n}n^{-2} \over \sum_k 2^{-k}2^{-2^k}k^{-2}}
\end{equation}
we have $D^{\mathrm{reg}}(\Lambda^\lambda_{\rho_p}\|\Lambda^\lambda_{\sigma_r}) < \infty$, while $\Dg(\Lambda^\lambda_{\rho_p}\|\Lambda^\lambda_{\sigma_r}) = \infty$ for all $\alpha \in (1,2]$ (and hence also for all $\alpha > 1$, as the geometric \renyi divergence is easily seen to be increasing in $\alpha$) and also $D_{\max}(\Lambda^\lambda_{\rho_p}\|\Lambda^\lambda_{\sigma_q}) = \infty$.

\section{Proofs of Main Results}
In this section we establish the necessary tools and then prove our infinite dimensional results. We start with a section on the geometric \renyi divergence in infinite dimensions, where we show that it satisfies a chain rule, similarly to what has previously been proven in finite dimensions. Secondly, we extend some inequalities involving quantum divergences to separable Hilbert spaces in \autoref{sec:InfExtensions}. These may be of independent interest. In \autoref{sec:InfQSL} we prove quantum Stein's lemma for in separable Hilbert spaces. and the one-shot result \autoref{thm:main} in \autoref{sec:ProofofMain}.
\subsection{The Geometric \renyi divergence in infinite dimensions}
\label{sec:geometric}
\newcommand{\hb}{\bar{h}_n}
\newcommand{\gb}{\bar{g}_n}
\newcommand{\Sf}{\widehat{S}_f}

Let us recall the definition of the geometric \renyi divergence as stated in the mathematical preliminaries section:
For $\rho, \sigma \in \P(\HS)$, we say that $(\Gamma, g, h, \mu)$ is a reverse test of $(\rho, \sigma)$, if $(X, \mu)$ is a $\sigma$-finite measure space, $g, h \in L^1(X, \mu)$ with $g, h \geq 0$, and $\Gamma$ is a positive trace-preserving map $\Gamma: L^1(X, \mu) \to \cB_1(\HS)$ such that $\Gamma(g) = \rho$, $\Gamma(h) = \sigma$. By $\Gamma$ being positive, we mean that $\Gamma(f) \geq 0$ for any $f \geq0$, and by $\Gamma$ being trace-preserving we mean that $\Tr(\Gamma(f)) = \int f d\mu$. 

\begin{definition}[\cite{hiai_quantum_2019}]
	In infinite dimensions, for $\rho, \sigma \in \P(\HS)$ and any $\alpha \in [0, \infty)$, the geometric \renyi trace function can be defined via the following optimization problem
	\begin{equation}
		\widehat{S}_{\alpha}(\rho\|\sigma) \coloneqq \min_{\Gamma, g, h, \mu}\{ S_{\alpha}^{\mu}(g\|h)\; |\; (\Gamma, g, h, \mu) \text{ is a reverse test of } (\rho, \sigma)\}
	\end{equation}
	where 
	\begin{equation}
		S_{\alpha}^{\mu}(g\|h) = \int h \br{g \over h}^\alpha d \mu\,,
	\end{equation}
	The geometric \renyi divergence is then defined as
	\begin{equation}
		\widehat{D}_{\alpha}(\rho\|\sigma) \coloneqq \frac{1}{\alpha - 1} \log \widehat{S}_\alpha(\rho\|\sigma)\,.
	\end{equation}
\end{definition}

This can be seen as a special case of what is called a maximal $f$-divergence for $f(\lambda) = \lambda^{\alpha}$. A lot is known about these divergences when $f$ is operator convex (e.g.~if $\alpha \in [1,2]$), in which case the solution to the optimization problem can be explicitly characterized \cite{hiai_quantum_2019}. 

Our main theorem of this section is the following:

\begin{restatethis}{theorem}{geometricChainRule}\label{thm:geometric_chain_rule}
	For all $\alpha \in [0,1) \cup (1,2]$, the geometric \renyi divergence satsifies the chain rule, i.e. for all states $\rho, \sigma \in \DM$ and any two channels $\E, \F : \P(\HS) \to \P(\mathcal{K})$ we have that
	\begin{equation}
		\Dg(\E(\rho)\|\F(\sigma)) \leq \Dg(\rho\|\sigma)+ \Dg(\E\|\F)\,.
	\end{equation}
	This also directly implies the addivity of the geometric \renyi channel divergence, i.e. for all channels $\E_1, \F_1: \P(\HS_1) \to \P(\mathcal{K}_1)$, $\E_2, \F_2: \P(\HS_{2}) \to \P(\mathcal{K}_2)$  it holds that
	\begin{equation}
	    \Dg(\E_1 \otimes \E_2 \| \F_1 \otimes \F_2) = \Dg(\E_1\|\F_1) + \Dg(\E_2\|\F_2).
	\end{equation}
\end{restatethis} 
We show this by proving that at least for some restricted set of states, the optimization in the reverse tests can be restricted to probability distributions on a countable set, rather than continuous $L^1$ functions (although the optimum might not necessarily be achieved anymore in this case), which then allows us to adapt the finite dimensional chain rule proof of \cite{berta_chain_2022} to show the desired result. To simplify the argument also in the case where $\alpha \in [0,1),$ we show this first step in the slightly more general setting of maximal $f$-divergences for an operator convex function $f$ defined as follows:
\begin{equation}
	\widehat{S}_f(\rho\|\sigma) \coloneqq \min_{\Gamma, g, h, \mu}\{ S_{f}^{\mu}(g\|h)\; |\; (\Gamma, g, h, \mu) \text{ is a reverse test of } (\rho, \sigma)\}
\end{equation}
where
\begin{equation}
	S_{f}^{\mu}(g\|h) = \int h f\br{g \over h} d \mu\,,
\end{equation}

Furthermore, for $\rho, \sigma \in \P(\HS)$, we say ($\Gamma$, $p$, $q$) is a \emph{discrete} reverse test for ($\rho$, $\sigma$), if $p, q \in \ell_1$, $p, q \geq 0$, and $\Gamma: \ell_1 \to \cB(\HS)$ is a linear, positive, and trace-preserving (in the sense of $\Tr(\Gamma(r)) = \sum_i r_i$ for all $r \in \ell_1$) map such that $\Gamma(p) = \rho$, $\Gamma(q) = \sigma$. 

\begin{lemma}\label{lem:discrete_reverse_tests}
	Let $\rho, \sigma \in \DM$ be such that $\exists c \geq 0 $ s.t. $\rho \leq c \sigma \leq c^2 \rho$. Let $f$ be a continuous operator convex function on $[0, \infty)$ for which there exists a finite function $\tilde{f}$, such that $\abs{f(\alpha x)} \leq \tilde{f}(\alpha) \abs{f(x)}$ for all $x, \alpha \in [0, \infty)$. Then, the maximal $f$ divergence can expressed as the following optimization problem:
	\begin{equation}
		\Sf(\rho\|\sigma) = \inf_{\Gamma, p, q} \{S_{f}(p\|q)\; |\; (\Gamma, p, q) \text{ is a discrete reverse test of } (\rho, \sigma) \}
	\end{equation}
	where 
	\begin{equation}
		S_f(p\|q) \coloneqq \sum_i q_i f\br{p_i \over q_i}
	\end{equation}
	
\end{lemma}

\begin{proof}
	For any measure $\mu$ on $[0,1]$, we say $(\Gamma, g, h, \mu)$ is a piecewise reverse test of $(\rho, \sigma)$ if $(\Gamma, g, h, \mu)$ is a reverse test, and additionally there exists a countable partition of $[0,1]$ into disjoint $\mu$-measurable sets $\{A_i\}$ of non-zero measure, such that 
	\begin{equation}\label{eq:piecewise_reverse_test}
		g = \sum_{i=1}^{\infty} g^{(i)} 1_{A_i} \qquad h = \sum_{i=1}^{\infty} h^{(i)} 1_{A_i} \,.
	\end{equation}
	where the $g^{(i)}$ and $h^{(i)}$ are constants.
	We first show the following statement:
	\begin{equation}\label{eq:div_with_piecewise}
		\Sf(\rho\|\sigma) = \inf_{\Gamma, g, h, \mu} \{S^{\mu}_{f}(g\|h)\; |\; (\Gamma, g, h, \mu) \text{ is a piecewise reverse test of } (\rho, \sigma)\}
	\end{equation}
	
	Since any piecewise reverse test is also a reverse test, we only have to show that there exists a sequence of piecewise reverse tests that converges to the optimum value. By \cite[Theorem 6.3]{hiai_quantum_2019}, since the function $f$ is operator convex, we can restrict to the case where the measure space $X$ is [0,1], and furthermore the optimum reverse test can be chosen as $g(t) = t$, $h(t) = 1-t$, $t \in [0,1]$, for some suitable $\Gamma$ and $\mu$ (see \cite{hiai_quantum_2019} for the exact expression of $\Gamma$ and $\mu$). 
	
	For $n \geq 2$, consider the following piecewise aproximation $g_n$ of $g(t) = t$: 
	\begin{equation}
		g_n(t) = \sum_{k = 1}^{\infty} 
		\frac{1}{n (k+1)} 1_{(\frac{1}{n(k+1)}, \frac{1}{n k }]}(t) +  \sum_{i = 1}^{n - 1}\frac{i}{n} 1_{(i/n, (i+1)/n]}(t)
	\end{equation}
	This satisfies the following properties:
	\begin{enumerate}
		\item $|g_n(t) - g(t)| \leq \frac{1}{n} \qquad \forall t \in [0,1]$
		\item $\frac{1}{2} g(t) \leq g_n(t) \leq g(t) \qquad \forall t \in [0,1]$\,.
	\end{enumerate}
	As $h(t) = g(1-t)$, we then set $h_n(t) \coloneqq g_n(1-t)$. See \autoref{fig:geometric_approximation} for an illustration of $g_n(t)$.
	
	\begin{figure}[htbp]
	    \centering
	    \includegraphics[width=\linewidth]{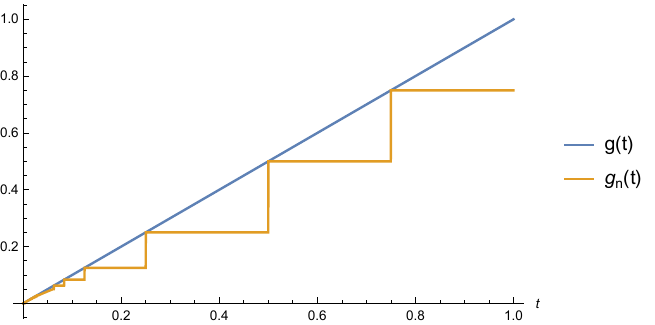}
	    \caption{Illustration of the piecewise approximation $g_n(t)$ of $g(t) = t$.}
	    \label{fig:geometric_approximation}
	\end{figure}
	
	The motivation for this approximation is that we want $g_n$ and $h_n$ to approximate $g$ and $h$ from below, but we also would like $g_n/h_n$ to be well-defined everywhere except for $t = 1$ (where also $g/h$ is infinite), hence the need for an infinite series at least in the definition of $h_n$. Note that the measurability of these functions follows from the measurability of $g$ and $h$, and (since $n \geq 2$) we also have $\frac{g_n}{h_n} \leq 2 \frac{g}{h}$. 
	These approximations will no longer form a reverse test, since for example in general $\Gamma(g_n) < \rho$, but the goal is now to modify these further to construct a $(\Gb, \gb, \hb)$ that does form a piecewise reverse test. Define
	\newcommand{\Dn}{\Delta^{(n)}}
	\begin{align}
		\Dn_g &\coloneqq \Gamma(g) - \Gamma(g_n) \\
		\Dn_h &\coloneqq \Gamma(h) - \Gamma(h_n)\,.
	\end{align}
	By the linearity and positivity of $\Gamma$, we have that these are positive operators. Moreover,
	\begin{equation}
		\Dn_g = \Gamma(g - g_n) \leq \Gamma\left(\frac{1}{n} 1_{[0,1]}\right) = \frac{1}{n} \Gamma(g + h) = \frac{1}{n}(\rho + \sigma) \leq \frac{1 + c}{n} \sigma \,.
	\end{equation}
	and similarly for $\Dn_h$ (remember that $g(t) = t$, $h(t) = 1-t$). Subsequently we will want to ensure that $\Dn_g \leq \Dn_h$.
	This will generally not be satisfied, but we can achieve this by replacing $h_n$ with $\tilde{h}_n = \left(1 - {1 \over \sqrt{n}}\right) h_n$. This is still a piecewise approximiation of $h$ that converges pointwise in the limit $n \to \infty$, but also
	\begin{equation}
		\Dn_{\tilde{h}} = \left(1 - {1 \over \sqrt{n}}\right) \Dn_h + {1 \over \sqrt{n}} \sigma \geq {1 \over \sqrt{n}} \sigma \geq \frac{\sqrt{n}}{1 + c} \Dn_g\,.
	\end{equation}
	For $n$ large enough, $\frac{\sqrt{n}}{1 + c} \geq 1$ and then $\Dn_g \leq \Dn_{\tilde{h}}$.
	We will construct our piecewise reverse test by extending $g_n$ and $h_n$ to functions on $[0,2]$ and setting suitable values on $(1,2]$ (this still proves \eqref{eq:div_with_piecewise}, since functions on $[0,1]$ and $[0,2]$ are obviously equivalent).
	As a measure on $[0,2]$, we choose $\nu \coloneqq \mu \oplus \lambda$, where $\lambda$ is the Lebesgue measure on $[1,2]$, and we mean by the notation that when integrating functions w.r.t.\ $\nu$, we integrate with $\mu$ over $[0,1]$ and with $\lambda$ over $[1,2]$, i.e.
	\begin{equation}
		\int_0^2 f d\nu = \int_0^1 f d \mu + \int_1^2 f d\lambda \,.
	\end{equation}
	Let
	\begin{equation}
		x \coloneqq \frac{\Tr(\Dn_g)}{\Tr(\Dn_{\tilde{h}})} \in [0,1]
	\end{equation}    
	and if $x < 1$, define the state $\omega$ by:
	\begin{equation}
		(1 - x) \omega \coloneqq \frac{\Dn_{\tilde{h}}}{\Tr(\Dn_{\tilde{h}})} - x \frac{\Dn_g}{\Tr(\Dn_g)} = \frac{1}{\Tr(\Dn_{\tilde{h}})} (\Dn_{\tilde{h}} - \Dn_g) \geq 0
	\end{equation}
	which satisfies $\Tr(\omega) = 1$. We employ the convention ${0 \over 0} = 0$, i.e.~if $\Dn_g = 0$, the term $(\Dn_g/ \Tr(\Dn_g))$ is zero. If $x = 1$, $\omega$ turns out to be irrelevant, so we can just set $\omega = 0$.
	Define further the following two normalized functions on $[0,2]$:
	\begin{align}
		\gb &= g_n 1_{[0,1]} + \frac{\Tr(\Dn_g)}{x} 1_{(1, 1+x]} = g_n 1_{[0,1]} + \Tr(\Dn_{\tilde{h}}) 1_{(1, 1+x]}  \\
		\hb &= \tilde{h}_n 1_{[0,1]} + \Tr(\Dn_{\tilde{h}}) 1_{(1,2]}
	\end{align}
	and the following map from $L^1$ functions on $[0,2]$ to bounded operators
	\begin{equation}
		\Gb(f) \coloneqq \Gamma(f 1_{[0,1]}) + \frac{\Dn_g}{\Tr(\Dn_g)} \int_{1}^{1 + x} f d \lambda + \omega \int_{1 + x}^{2} f d\lambda \,.
	\end{equation}
	This is positive, in the sense that $\Gb(f) \geq 0$ if $f \geq 0$, and normalization-preserving in the sense that
	\begin{equation}
		\Tr(\Gb(f)) = \int_{0}^2 f d\nu \,.
	\end{equation}
	It also satisfies
	\begin{align}
		\Gb(\gb) &= \Gamma(g_n) + \Dn_g x \frac{\Tr(\Dn_{\tilde{h}})}{\Tr(\Dn_g)} =  \Gamma(g_n) + \Dn_g = \rho \\
		\Gb(\hb) &= \Gamma(\tilde{h}_n) + \Dn_g x \frac{\Tr(\Dn_{\tilde{h}})}{\Tr(\Dn_g)} + (1-x) \omega \Tr(\Dn_{\tilde{h}})  = \Gamma(\tilde{h}_n) + \Tr(\Dn_{\tilde{h}}) \frac{\Dn_{\tilde{h}}}{\Tr(\Dn_{\tilde{h}})} = \sigma
	\end{align}
	and hence $(\Gb, \gb, \hb, \nu)$ is a piecewise reverse test of $(\rho, \sigma)$. It remains to show that
	\begin{equation}
		S^{\nu}_{f}(\gb\|\hb) \xrightarrow{n \to \infty} S_{f}^{\mu}(g\|h) \,.
	\end{equation}        
	On $(1,2]$, we have
	\begin{equation}
		\int_1^2 \abs{\hb f\br{\gb \over \hb}} d \lambda = \int_1^2 \Tr(\Dn_{\tilde{h}}) \abs{f\br{\Tr(\Dn_{\tilde{h}}) 1_{[1, 1+x]} \over \Tr(\Dn_{\tilde{h}})}} d \lambda \leq C \Tr(\Dn_{\tilde{h}}) = C \int_0^1 (h - \tilde{h}_n) d \mu \to 0
	\end{equation}
	by monotone convergence, where $C = \max\{\abs{f(0)}, \abs{f(1)}\} < \infty$. On $[0,1]$ we have
	\begin{equation}
		{g_n \over \tilde{h}_n} \leq \frac{2}{1-{1 \over \sqrt{n}}} \frac{g}{h} \leq 4 \frac{g}{h}
	\end{equation}
	and thus by the assumption on $f$:
	\begin{equation}
		 \abs{\tilde{h}_n f\br{\frac{g_n}{\tilde{h}_n}}}  \leq  \tilde{f}(4) \abs{h f\br{\frac{g}{h}}} \,.
	\end{equation}
	If $S^{\mu}_{f}(g\|h) = \infty$ there is nothing to show, otherwise by dominated convergence (and continuity of $f$)
	\begin{equation}
		\int_0^1 \tilde{h}_n f \br{\frac{g_n}{\tilde{h}_n}} d \mu \to \int_0^1 h f\br{\frac{g}{h}} d \mu = S^{\mu}_{f}(g\|h)
	\end{equation}
	
	This completes the proof of \eqref{eq:div_with_piecewise}. 
	Finally, we still have to show that the formulation in terms of piecewise reverse tests is equivalent to the formulation in terms of discrete reverse tests. Let ($\Gamma$, $p$, $q$) be a discrete reverse test of $(\rho, \sigma)$. Then let $\mu$ be the Lebesgue measure on $[0,1]$, and $\{A_i\}_{i = 1}^{\infty}$ be a countable collection of disjoint measurable subsets of $[0,1]$, such that $\mu(A_i) > 0$, for all $i$ and $\sum_i \mu(A_i) = 1$. Define $g' \coloneqq \sum_i \frac{1_{A_i}}{\mu(A_i)} p_i$, $h' \coloneqq \sum_i \frac{1_{A_i}}{\mu(A_i)} q_i$ and $\Gamma'(f) \coloneqq \sum_{i} \Gamma_i \int_{A_i} f d\mu$ for any $f \in L^1$. Then, $\Gamma'(g') = \Gamma(p) = \rho$ and $\Gamma'(h') = \Gamma(q) = \sigma$, and so $(\Gamma', g', h', \mu)$ is a piecewise reverse test, and additionally also $S_{\alpha}(p\|q) = S_{f}^{\mu}(g'\|h')$.
	Conversely, if $(\Gamma', g', h', \mu)$ is a piecewise reverse test as in \eqref{eq:piecewise_reverse_test}, then define $p, q \in \ell_1$ by $p_i \coloneqq g'^{(i)} \mu(A_i)$, $q_i \coloneqq h'^{(i)} \mu(A_i)$, and $\Gamma_i \coloneqq \frac{1}{\mu(A_i)}\Gamma'(1_{A_i})$, which defines a discrete reverse test again achieving the same value in the optimization problem.
	
\end{proof}
\begin{lemma}\label{lem:geometric_scalar}
	For any $\alpha \in [0, \infty)$, any states $\rho, \sigma \in \P(\HS_A)$, and any two positive real numbers $p, q$ the geometric trace function satisfies
	\begin{equation}
		\Sg(p \rho\|q \sigma) = p^{\alpha} q^{1 - \alpha} \Sg(\rho\|\sigma).
	\end{equation}
\end{lemma}
\begin{proof}
	For any two functions $g$ and $h$ and any measure $\mu$, it is easy to see that 
	\begin{equation}
		\Sg^{\mu}(p g\|q h) = p^{\alpha} q^{1 - \alpha} \Sg^{\mu}(g\|h)
	\end{equation}
	The statement then follows from the variational expression of $\Sg$ and the fact that if $(\Gamma, g, h, \mu)$ forms a reverse test of $(\rho, \sigma)$, then $(\Gamma, p g, p h, \mu)$ form a reverse test of $(p \rho, q \sigma)$ and vice versa. 
\end{proof}

\begin{proof}[Proof of \autoref{thm:geometric_chain_rule}]
	Let $\rho = \rho_A, \sigma = \sigma_A \in \DM[\HS_A]$ and fix $\varepsilon > 0$. Then, define $\tilde{\rho}_A = \rho_A + \varepsilon(\rho_A + \sigma_A)$, $\tilde{\sigma}_A = \sigma_A + \varepsilon(\rho_A + \sigma_A)$. 
	Let ($\Gamma$, $p$, $q$) be a discrete reverse test of $(\tilde{\rho}_A, \tilde{\sigma}_A)$.
	We write $\Gamma_i = \Gamma(1_i) \in \cB_1$, where $1_i \in \ell_1$ is the sequence that is one at index $i$ and zero otherwise. We then have 
	\begin{equation}
		\tilde{\rho}_A = \sum_i p_i \Gamma_i \qquad \tilde{\sigma}_A = \sum_i q_i \Gamma_i \,.
	\end{equation}
	Hence, taking $R$ an additional infinite dimensional system with countable basis $\{\ket{i}\}_i$, we can define
	\begin{equation}
		\tilde{\rho}_{RA}= \sum_i p_i \ketbra{i}{i}_R \otimes \Gamma_i \qquad \tilde{\sigma}_{RA} = \sum_i q_i \ketbra{i}{i}_R \otimes \Gamma_i
	\end{equation}
	such that $\tilde{\rho}_A = \Tr_R(\tilde{\rho}_{RA}), \tilde{\sigma}_A = \Tr_R(\tilde{\sigma}_{RA})$. We start with the case where $\alpha \in (1,2]$, for which it is well-known that the function $f(x) = x^{\alpha}$ is operator convex. 
	We have,
	\begin{align}
		\Sg(\E(\tilde{\rho}_A)\|\F(\tilde{\sigma}_A)) &\leq
		\Sg(\E(\tilde{\rho}_{RA})\|\F(\tilde{\sigma}_{RA})) \\ &= \Sg\left(\sum_i  p_i \ketbra{i}{i} \otimes \E(\Gamma_i)\|\sum_i q_i \ketbra{i}{i} \otimes \F(\Gamma_i)\right) \\ &= \sum_i \Sg(p_i \E(\Gamma_i)\|q_i \F(\Gamma_i)) \\ &= \sum_i p_i^{\alpha} q_i^{1 - \alpha} \Sg(\E(\Gamma_i)\|\F(\Gamma_i)) \\
		&\leq \sum_i p_i^{\alpha} q_i^{1 - \alpha} \sup_{\nu \in \DM[A]} \Sg(\E(\nu)\|\F(\nu)) \\
		&= \Sg(p\|q) \sup_{\nu \in \DM[A]} \Sg(\E(\nu)\|\F(\nu))
	\end{align}
	where the first inequality follows from the data-processing inequality for the geometric trace function \cite[Theorem 2.9]{hiai_quantum_2019} (remember that the channels act only on system $A$), the third line follows from \autoref{lem:generalized_direct_sum_decomp}, and the fourth equality is \autoref{lem:geometric_scalar}. 
	
	We write $\Sg(\E\|\F) = \sup_{\nu \in \DM[A]} \Sg(\E(\nu)\|\F(\nu))$. Taking the infimum over all discrete reverse tests, we find by \autoref{lem:discrete_reverse_tests} that
	\begin{equation}
		\Sg(\E(\tilde{\rho})\|\F(\tilde{\sigma})) \leq \Sg(\tilde{\rho}\|\tilde{\sigma}) \Sg(\E\|\F)\,.
	\end{equation}
	It remains to show convergence in the limit $\varepsilon \to 0$. For the right-hand side, by \cite[Definition 2.8]{hiai_quantum_2019}
	\begin{equation}
		\lim_{\varepsilon \to 0} \Sg(\tilde{\rho}\|\tilde{\sigma}) = \Sg(\rho\|\sigma) \,.
	\end{equation}
	For the left-hand side, by \cite[Theorem 2.9]{hiai_quantum_2019}, we have
	\begin{align}
		\Sg(\E(\tilde{\rho})\|\F(\tilde{\sigma})) &= \Sg(\E(\rho) + \varepsilon \E(\rho + \sigma))\|\F(\sigma) + \varepsilon \F(\rho + \sigma)) \\
		&\leq 
		\Sg(\E(\rho)\|\F(\sigma)) + \varepsilon \Sg(\E(\rho + \sigma)\|\F(\rho + \sigma)) \\&\leq \Sg(\E(\rho)\|\F(\sigma)) + \varepsilon \Sg(\E\|\F).
	\end{align}
	Now, if $\Sg(\E\|\F) = \infty$ the statement of our theorem is empty, so we can assume it to be finite. In that case we get 
	\begin{equation}
		\limsup_{\varepsilon \to 0} \Sg(\E(\tilde{\rho})\|\F(\tilde{\sigma})) \leq \Sg(\E(\rho)\|\F(\sigma))\,.
	\end{equation}
	The opposite direction
	\begin{equation}
		\liminf_{\varepsilon \to 0} \Sg(\E(\tilde{\rho})\|\F(\tilde{\sigma})) \geq \Sg(\E(\rho)\|\F(\sigma))
	\end{equation}
	follows by the lower semi-continuity of $\Sg$ \cite[Theorem 5.5]{hiai_quantum_2019}. The statement then follows upon taking the logarithm and dividing by $\alpha - 1$ (which is positive).

	For $\alpha \in [0, 1)$, $f(x) = - x^{\alpha}$ is well-known to be operator convex, and in that case we can apply all the above reasoning to $\Sf = - \Sg$ (all the properties from \cite{hiai_quantum_2019} we used apply to $\Sf$ with an operator convex function $f$). We then find that
	\begin{equation}
		\Sf(\E(\rho_A)\|\F(\sigma_A)) \leq \Sg(\rho\|\sigma) \sup_{\nu \in \DM[A]} \Sf(\E(\nu)\|\F(\nu))
	\end{equation}
	which corresponds to 	
	\begin{equation}
		\Sg(\E(\rho_A)\|\F(\sigma_A)) \geq \Sg(\rho\|\sigma) \inf_{\nu \in \DM[A]} \Sg(\E(\nu)\|\F(\nu))\,.
	\end{equation}
	The desired statement follows again after taking logarithms and dividing by $\alpha - 1$ (which is now negative, so it turns around the inequality and changes the infimum into a supremum). 
	\medskip
	
	To see how the chain rule implies additivity, let $\E_1, \F_1: \P(\HS_{A_1}) \to \P(\HS_{B_1})$, $\E_2, \F_2: \P(\HS_{A_2}) \to \P(\HS_{B_2})$ be channels, and consider any joint input state $\nu = \nu_{R A_1 A_2}$. Then (supressing identities as before)
	\begin{align}
	    \Dg\big((\E_1 \otimes \E_2)(\nu)\|(\F_1 \otimes \F_2)(\nu)\big) &= \Dg\big(\E_1(\E_2(\nu))\|\F_1(\F_2(\nu)\big) \\
	    &\leq \Dg(\E_1\|\F_1) + \Dg(\E_2(\nu)\|\F_2(\nu)) \\
	    &\leq \Dg(\E_1\|\F_1) + \Dg(\E_2\|\F_2)
	\end{align}
	where we used the chain rule in the first inequality. This implies
	\begin{equation}
	    \Dg(\E_1\otimes \E_2\|\F_1 \otimes \F_2) \leq \Dg(\E_1\|\F_1) + \Dg(\E_2\|\F_2),
	\end{equation}
	and the other direction follows by just restricting the supremum over input states $\nu_{R A_1 A_2}$ to product states.
\end{proof}

\begin{remark}[Chain Rule for the Belavkin-Staszewski Relative Entropy]
	By using a very similar argument as in \autoref{thm:geometric_chain_rule}, one can also show that the Belavkin-Staszewski relative entropy, which is defined as the maximal $f$-divergence for the operator convex function $f(x) = x \log x $, satisfies the chain rule. To see this, note that now instead of \autoref{lem:geometric_scalar}, we have $\sum_i \Sf(p_i \rho_i\|q_i \sigma_i) = \Sf(p\|q) + \sum_i p_i \Sf(\rho_i\|\sigma_i)$, and the $p_i$ are normalized. The remainder of the argument is then almost identical to \autoref{thm:geometric_chain_rule}. 
\end{remark}

\subsection{Some infinite dimensional relative entropy inequalities}\label{sec:InfExtensions}
In this section we prove the infinite dimensional version of the inequalities \eqref{equ:ImpIneq1}-\eqref{equ:ImpIneq5}. We start with the following Lemma
\begin{lemma}
\label{lem:3.1}
For two states $\rho,\sigma\in\mathcal{D}(\mathcal{H})$ on a separable Hilbert space and for any $\lambda\in(-\infty,D_{\max}(\rho\|\sigma)]$, if $\epsilon = \Tr\Sigma$, where $\Sigma:=(\rho-2^\lambda\sigma)_+$, then
\begin{equation}
    D_{\max}^{\sqrt{\epsilon}}(\rho\|\sigma)\leq \lambda-\log(1-\epsilon).
\end{equation}
\smallskip

\noindent
\textbf{Remark.} This lemma is similar to \cite[Lemma 15]{art:DattaMaxRelEnt}. However, there the definition of the smoothed max-relative entropy was with respect to the trace-norm ball of sub-normalized states, whereas in this work we smooth over the Sine-distance (or purified distance) ball of normalized states, which yields a quantitatively different result. 
\end{lemma}
\begin{proof}
The proof is similar to \cite[Lemma 6.21]{art:Tomamichel_book}. Let $\Lambda:=2^\lambda\sigma$ and $\Sigma:=(\rho-2^\lambda\sigma)_+$ and define $G:=\Lambda^{\frac{1}{2}}(\Lambda+\Sigma)^{-\frac{1}{2}}$. This is well-defined since $G$ and $G^*G$ are contractions, as
\begin{align*}
    &G^*G=(\Lambda+\Sigma)^{-\frac{1}{2}}\Lambda(\Lambda+\Sigma)^{-\frac{1}{2}}\overset{\Lambda\leq\Lambda+\Sigma}{\leq}(\Lambda+\Sigma)^{-\frac{1}{2}}(\Lambda+\Sigma)(\Lambda+\Sigma)^{-\frac{1}{2}} = \1, \\
    &\implies \|G^*G\|=\|GG^*\|=\|G\|^2\leq 1.
\end{align*}
Now observe that
\begin{align*}
    \rho&\leq \Lambda+\Sigma = 2^\lambda\sigma+(\rho-2^\lambda\sigma)_+ \Longleftrightarrow \\ 0 &\leq (2^\lambda\sigma-\rho)+(\rho-2^\lambda\sigma)_+ = (2^\lambda\sigma-\rho)_+-(2^\lambda\sigma-\rho)_-+(2^\lambda\sigma-\rho)_- = (2^\lambda\sigma-\rho)_+,
\end{align*}
which is obviously true. Thus it follows that
\begin{align*}
    G\rho G^*\leq G(\Lambda+\Sigma)G^*=\Lambda^{\frac{1}{2}}(\Lambda+\Sigma)^{-\frac{1}{2}}(\Lambda+\Sigma)(\Lambda+\Sigma)^{-\frac{1}{2}}\Lambda^\frac{1}{2} = \Lambda = 2^\lambda\sigma,
\end{align*} and
\begin{align*}
    1-\Tr[G^*G\rho]&=\Tr[(\1-G^*G)\rho]\overset{\rho\leq\Lambda+\Sigma}{\leq}\Tr[(\1-G^*G)(\Lambda+\Sigma)] \\ &= \Tr[\Lambda+\Sigma]-\Tr[(\Lambda+\Sigma)^{-\frac{1}{2}}\Lambda(\Lambda+\Sigma)^{-\frac{1}{2}}(\Lambda+\Sigma)]=\Tr[\Sigma] \\
    \Longleftrightarrow \Tr[G^*G\rho]&\geq 1-\Tr[\Sigma]= 1-\epsilon
\end{align*}
So defining the state $\tilde{\rho}:=\frac{G\rho G^*}{\Tr[G^*G\rho]}\in\mathcal{D}(\mathcal{H})$, we have by the above that
\begin{align*}
    \tilde{\rho}\leq\frac{2^\lambda}{\Tr[G^*G\rho]}\sigma \leq \frac{2^\lambda}{1-\epsilon}\sigma =2^{\lambda-\log(1-\epsilon)}\sigma
\end{align*} and so by the Gentle Measurement Lemma (\autoref{lem:GentleMeasurement}) it holds that $P(\rho,\tilde{\rho})\leq \sqrt{\Tr[(\1-G^*G)\rho]}\leq \sqrt{\Tr[\Sigma]}=\sqrt{\epsilon}$. Thus the desired result follows immediately:
\begin{align*}
    D_{\max}^{\sqrt{\epsilon}}(\rho\|\sigma)\leq D_{\max}(\tilde{\rho}\|\sigma) \leq \lambda-\log(1-\epsilon)
\end{align*}
\end{proof} Note, that we could have equally have chosen $\epsilon=1-\Tr[G^*G\rho]$, which would give a numerically slightly tighter bound in the Lemma. What we chose  though suffices for the argument we are making and makes some calculations simpler later on.
Now, as a simple Corollary we get the upper bounds on the smoothed max divergence in terms of the $\alpha$-Petz-\renyi and the hypothesis testing divergences.
\begin{corollary}
For any $\epsilon\in(0,1)$ and two states $\rho,\sigma\in\mathcal{D}(\mathcal{H})$ on a separable Hilbert space $\mathcal{H}$ the following inequality holds:
\begin{equation}
    D_{\max}^\epsilon(\rho\|\sigma)\leq D_H^{1-\epsilon^2}(\rho\|\sigma)-\log(1-\epsilon^2)
    \label{equ:ImpIneq1NEW}
\end{equation}
\label{COR:1}
\end{corollary}
\begin{proof}
Fix some $\lambda$ and let $\epsilon=\Tr[\Sigma]$ be as in \autoref{lem:3.1} above w.r.t this $\lambda$.
Let $P_+ = \{\rho-2^\lambda\sigma\}_+$ be the projector onto the support of $(\rho-2^\lambda\sigma)_+$. Now consider the POVM $\{P_+,P_-\}$ as the decision rule for a hypothesis test between $\rho$ and $\sigma$. The associated type I and II errors are, respectively, 
\begin{align*}
    \alpha &= \Tr[P_-\rho] = 1-\Tr[P_+\rho] \overset{2^\lambda\sigma\geq0}{\leq} 1-\Tr[P_+(\rho-2^\lambda\sigma)] = 1-\Tr[\Sigma] = 1-\epsilon \\
    \beta &= \Tr[P_+\sigma]\underset{P_+(\rho-2^\lambda\sigma)\geq 0}{\leq}2^{-\lambda}\Tr[P_+\rho]\leq 2^{-\lambda} 
\end{align*}
Therefore it follows that
\begin{align*}
    D_H^{1-\epsilon}(\rho\|\sigma)=-\log\inf_{0\leq F\leq\1}\left\{\Tr[F\sigma]|\Tr[(\1-F)\rho]\leq 1-\epsilon \right\}\geq -\log2^{-\lambda} = \lambda.
\end{align*}
Thus the claim follows immediately from \autoref{lem:3.1} when substituting $\epsilon$ by $\epsilon^2$:
\begin{align*}
    D_{\max}^{\sqrt{\epsilon}}(\rho\|\sigma)\leq \lambda-\log(1-\epsilon) \leq D_H^{1-\epsilon}(\rho\|\sigma) +\log\left(\frac{1}{1-\epsilon}\right).
\end{align*}
\end{proof}
\begin{corollary}
For any $\epsilon\in(0,1)$ and $\alpha\in(1,\infty)$ and two states $\rho,\sigma\in\mathcal{D}(\mathcal{H})$ on a separable Hilbert space $\mathcal{H}$ the following inequality holds:
\begin{equation}
    D_{\max}^\epsilon(\rho\|\sigma)\leq \mathbb{D}_\alpha(\rho\|\sigma)+\frac{2}{\alpha-1}\log\left(\frac{1}{\epsilon}\right)+\log\left(\frac{1}{1-\epsilon^2}\right).
\label{equ:ImpIneq2NEW}
\end{equation}
where $\mathbb{D}_\alpha$ can be any quantum $\alpha$-\renyi divergence, i.e.\ any function on quantum states that satisfies the data-processing inequality, and reduces to the classical $\alpha$-\renyi divergence when evaluated on commuting states. Specifically, the result will hold for the Petz-\renyi divergence.
\label{COR:2}
\end{corollary}
\begin{proof}
The proof is the same as in \cite[Proposition 6.22]{art:Tomamichel_book}. Fix some $\lambda$ and let $\epsilon=\Tr[\Sigma]$ be as in \autoref{lem:3.1} above w.r.t this $\lambda$.
Denote the spectral measure of the compact operator $(\rho-2^\lambda\sigma)$ with $\{\ketbra{\nu_i}{\nu_i}\}_{i\in S}$ and set $S_+:=\{i\in S| \langle\nu_i|\rho-2^\lambda\sigma|\nu_i\rangle\geq 0\}$. Set $p_i:=\langle\nu_i|\rho|\nu_i\rangle$ and $q_i:=\langle\nu_i|\sigma|\nu_i\rangle$, then $P:=\{p_i\}_{i\in S}$ and $Q:=\{q_i\}_{i\in S}$ are  probability measures on $S$. Now if $i\in S_+$, then $p_i-2^\lambda q_i\geq 0 \Leftrightarrow \frac{p_i}{q_i}2^{-\lambda}\geq 1$. Let $\alpha\in(1,\infty)$. Now we have that
\begin{align*}
    \epsilon = \Tr[\Sigma] &= \sum_{i\in S_+}p_i-2^\lambda q_i\leq \sum_{i\in S_+}p_i \leq \sum_{i\in S_+}p_i\left(\frac{p_i}{q_i}2^{-\lambda}\right)^{\alpha-1} \\ &= 2^{\lambda(1-\alpha)}\sum_{i\in S_+}p_i^\alpha q_i^{1-\alpha} \leq 2^{\lambda(1-\alpha)}\sum_{i\in S}p_i^\alpha q_i^{1-\alpha}
\end{align*}
\begin{align*}
 &\implies  \log\epsilon \leq -\lambda(\alpha-1)+\log\left(\sum_{i\in S}p_i^\alpha q_i^{1-\alpha}\right) \\
 &\Longleftrightarrow \lambda(\alpha-1)\leq \log\left(\sum_{i\in S}p_i^\alpha q_i^{1-\alpha}\right) - \log\epsilon \\
 &\Longleftrightarrow \lambda \leq \frac{1}{\alpha-1}\log\left(\sum_{i\in S}p_i^\alpha q_i^{1-\alpha}\right)+\frac{1}{\alpha-1}\log\left(\frac{1}{\epsilon}\right)\leq D_\alpha(P\|Q)+\frac{1}{\alpha-1}\log\left(\frac{1}{\epsilon}\right).
\end{align*}
Where $D_\alpha(P\|Q)$ is the classical $\alpha$-Renyi divergence. It is upper bounded by any quantum $\alpha$-\renyi divergence $\mathbb{D}_\alpha(\rho\|\sigma)$, via the data processing inequality applied to the channel that implements a measurement in the $\{\ketbra{\nu_i}{\nu_i}\}_i$ basis. With the previous \autoref{lem:3.1} we thus have
\begin{align*}
    D_{\max}^{\sqrt{\epsilon}}(\rho\|\sigma)\leq \lambda+\log\left(\frac{1}{1-\epsilon}\right) \leq \mathbb{D}_\alpha(\rho\|\sigma) +\log\left(\frac{1}{1-\epsilon}\right)+\frac{1}{\alpha-1}\log\left(\frac{1}{\epsilon}\right).
\end{align*}
The result then follows when replacing $\epsilon$ with $\epsilon^2$.
\end{proof}

\begin{lemma}[\textit{Proposition 4} from \cite{wang_resource_states_2019}]
\label{lemma:SmoothedDmaxLowerBoundWW}
For any two states $\rho,\sigma \in\mathcal{D}(\mathcal{H})$ on a separable Hilbert space $\mathcal{H}$ and any $\epsilon\in(0,1),\alpha\in[0,1)$ it holds that
\begin{equation}
    D_{\max}^\epsilon(\rho\|\sigma) \geq D_\alpha(\rho\|\sigma) + \frac{2}{\alpha-1}\log\left(\frac{1}{1-\epsilon}\right).
    \label{equ:ImpINeq3REP}
\end{equation}
\begin{proof} 
The proof of this statement in the finite dimensional case from \cite{wang_resource_states_2019} also holds in the infinite dimensional case. For the convenience of the reader it is repeated here. \\
\textit{Claim 1:} If $\rho_0,\rho_1,\sigma\in\mathcal{D}(\mathcal{H})$ are s.t. supp($\rho_{1})\subset\text{supp}(\sigma)$, $\alpha\in(0,1)$, and $\beta:=2-\alpha\in(1,2)$,
then it holds that 
\begin{equation}
    D_\beta(\rho_0\|\sigma)-D_\alpha(\rho_1\|\sigma) \geq \frac{2}{1-\alpha}\log\left(1-\frac{1}{2}\|\rho_0-\rho_1\|_1\right)\geq \frac{2}{1-\alpha}\log(1-P(\rho_0,\rho_1)).
\end{equation}
\textit{Proof of claim 1}: By definition $\alpha-1=1-\beta$, so that
\begin{align*}
    D_\beta(\rho_0\|\sigma)-D_\alpha(\rho_1\|\sigma) &= \frac{1}{\beta-1}\log\Tr[\rho_0^\beta\sigma^{1-\beta}]-\frac{1}{\alpha-1}\log\Tr[\rho_0^\alpha\sigma^{1-\alpha}] \\
    &=\frac{1}{\beta-1}\log\{\Tr[\rho_0^\beta\sigma^{1-\beta}]\Tr[\rho_0^\alpha\sigma^{1-\alpha}]\} \\
    &=\frac{1}{\beta-1}\log\{\|\rho_0^{\frac{\beta}{2}}\sigma^{\frac{1-\beta}{2}}\|^2_2\|\sigma^{\frac{1-\alpha}{2}}\rho_1^{\frac{\alpha}{2}}\|^2_2\} \\ &\overset{C.S.}{\geq} \frac{1}{\beta-1}\log\{\|\rho_0^{\frac{\beta}{2}}\sigma^{\frac{1-\beta}{2}}\sigma^{\frac{1-\alpha}{2}}\rho_1^{\frac{\alpha}{2}}\|^2_1\} = \frac{2}{\beta-1}\log\|\rho_0^{\frac{\beta}{2}}\rho_1^{1-\frac{\beta}{2}}\|_1 \\ &\geq \frac{2}{\beta-1}\log\Tr[\rho_0^{\frac{\beta}{2}}\rho_1^{1-\frac{\beta}{2}}] \\ &\overset{\text{\cite{art:DiscriminatingStates:QCB}}}{\geq} \frac{2}{\beta-1}\log\Tr[\frac{1}{2}(\rho_0+\rho_1-|\rho_0-\rho_1|)] \\ &= \frac{2}{\beta-1}\log\Tr[1-\frac{1}{2}\|\rho_0-\rho_1\|_1] \geq \frac{2}{\beta-1}\log\Tr[1-P(\rho_0,\rho_1)].
\end{align*}
The first inequality follows from Cauchy-Schwarz inequality (C.S.), the second from \cite[Theorem 1]{art:DiscriminatingStates:QCB}, which is applicable since $\frac{\beta}{2}\in(\frac{1}{2},1)\subset[0,1]$ and which holds in infinite dimensions. The last inequality follows by the Fuchs-van de Graaf inequality, \autoref{lem:FuchsVanDeGraaf}. 
\\ \textit{Claim 2:} For any $\rho,\sigma\in\mathcal{D}(\mathcal{H})$ it holds that $D_{\max}(\rho\|\sigma)\geq D_2(\rho\|\sigma)$. \\
\textit{Proof of Claim 2}: 
\begin{align*}
   D_2(\rho\|\sigma)&=\log\Tr[\rho^2\sigma^{-1}]=\log\Tr[\rho\rho^{\frac{1}{2}}\sigma^{-1}\rho^{\frac{1}{2}}] \\ &\leq\log\sup_{\|\tau\|_1\leq1}\Tr[\tau\rho^{\frac{1}{2}}\sigma^{-1}\rho^{\frac{1}{2}}] = \log\|\rho^{\frac{1}{2}}\sigma^{-1}\rho^{\frac{1}{2}}\| = D_{\max}(\rho\|\sigma)
\end{align*}
Now to proof the statement of the Lemma, fix $\alpha\in(0,1)$ and $\tilde{\rho}\in B^\epsilon(\rho)$, then for $\beta=2-\alpha\in(1,2)$ we have, by the second claim and the monotonicity of $\alpha\mapsto D_\alpha(\rho\|\sigma)$ on (0,1), that
\begin{align*}
    D_{\max}(\tilde{\rho}\|\sigma) &\geq D_\beta(\tilde{\rho}\|\sigma) \overset{\text{claim 1}}{\geq} D_\alpha(\rho\|\sigma) + \frac{2}{1-\alpha}\log(1-P(\rho,\tilde{\rho})) \\ &\geq D_\alpha(\rho\|\sigma) + \frac{2}{\alpha-1}\log\left(\frac{1}{1-\epsilon}\right).
\end{align*}
Optimizing over all $\tilde{\rho}\in B^\epsilon(\rho)$ yields the desired statement for $\alpha\in(0,1)$. Taking the limit $\alpha\to0$ gives the statement for $\alpha=0$.
\end{proof}
\end{lemma}

\subsection{Quantum Stein's lemma for channels in infinite dimensions (\autoref{thm:steins})} \label{sec:InfQSL}
One of the most fundamental questions of hypothesis testing is to study asymptotic error decay rates. In asymmetric settings these asymptotic error decay rates often turn out to be expressible in terms of relative entropies, and the statements proving such expressions are usually called Stein's lemmas, in honour of Stein's original result for classical binary asymmetric hypothesis testing.

The parallel quantum Stein's lemma for channels between \emph{finite dimensional} Hilbert spaces, recently proven by Wang and Wilde \cite[Theorem 3]{wang_resource_2019}, is the statement that in the asymmetric setting the optimal asymptotic decay rate of discriminating between two channels $\mathcal{E},\mathcal{F}$ using a parallel strategy\footnote{It also holds for adaptive strategies, since there the optimal decay rate is given by the amortized channel divergence, which was recently shown to be equal to the regularized in \cite{fang_chain_2020}.} in the limit of vanishing type I error $\epsilon\to0$ is given by the regularized Umegaki-channel divergence between these two channels, i.e.
\begin{equation}
    \lim_{\epsilon\to0}\lim_{n\to\infty}\frac{1}{n}D_H^\epsilon(\mathcal{E}^{\otimes n}\|\mathcal{F}^{\otimes n})=D^{\text{reg}}(\mathcal{E}\|\mathcal{F}) = \lim_{n\to\infty}\frac{1}{n}D(\mathcal{E}^{\otimes n}\|\mathcal{F}^{\otimes n}).
    \label{equ:QuantumChannelStein}
\end{equation}
It was proven by separately upper \cite{wang_resource_2019,wang_one-shot_2012} and lower bounding \cite{wang_resource_2019,art:SecondOrderAsymptot,art:AHirarchyOfInformation6574274} the l.h.s of \eqref{equ:QuantumChannelStein} by the regularized channel divergence. The upper bound relies on inequality \eqref{equ:ImpIneq4} which holds in infinite dimensional Hilbert spaces. Its proof is repeated for the convenience of the reader in \autoref{app:CQSLUpperBound}. The lower bound in \cite{wang_resource_2019}, however, does not evidently hold in the infinite dimensional case. The following proof we give based on inequalities \eqref{equ:ImpIneq1NEW} and \eqref{equ:ImpIneq3}, though, does.
\begin{proposition} \label{prop:parallel_stein}
Let $\mathcal{E},\mathcal{F}:\P(\mathcal{H})\to\P(\mathcal{K})$ be two quantum channels (CPTP maps) where $\mathcal{H}, \mathcal{K}$ are arbitrary separable Hilbert spaces. Then
\begin{equation}
    \lim_{\epsilon\to0}\liminf_{n\to\infty}\frac{1}{n}D_H^\epsilon(\mathcal{E}^{\otimes n}\|\mathcal{F}^{\otimes n})\geq D^\mathrm{reg}(\mathcal{E}\|\mathcal{F}).
\end{equation}
\end{proposition}
\begin{proof}
Combining inequalities \eqref{equ:ImpIneq1NEW} and \eqref{equ:ImpINeq3REP} and rearranging we get
\begin{align*}
    D_H^\epsilon(\rho\|\sigma) \geq D_\alpha(\rho\|\sigma)-\log\left(\frac{1}{\epsilon}\right)+\frac{2}{\alpha-1}\log\left(\frac{1}{1-\sqrt{1-\epsilon}}\right)
\end{align*} for any $\alpha\in(0,1)$ and for any states $\rho,\sigma\in\mathcal{D}(\mathcal{H})$.
Since the $\alpha$-Petz-\renyi divergence is additive, i.e. $D_\alpha(\rho\otimes\omega\|\sigma\otimes\tau)= D_\alpha(\rho\|\sigma)+D_\alpha(\omega\|\tau)$ it follows that
$\frac{1}{n}D_H^\epsilon(\rho^{\otimes n}\|\sigma^{\otimes n})\geq D_\alpha(\rho\|\sigma)+\frac{1}{n}\left(\frac{2}{\alpha-1}\log\left(\frac{1}{1-\sqrt{1-\epsilon}}\right)-\log\left(\frac{1}{\epsilon}\right)\right)$ holds for any $\alpha\in(0,1)$.
Thus we have
\begin{align*}
    &\frac{1}{nm}D_H^\epsilon(\mathcal{E}^{\otimes nm}\|\mathcal{F}^{\otimes nm}) = \frac{1}{nm}\sup_{\rho_{RA^{nm}}}D_H^\epsilon(\mathcal{E}^{\otimes nm}(\rho_{RA^{nm}})\|\mathcal{F}^{\otimes nm}(\rho_{RA^{nm}})) \\
     \geq &\frac{1}{nm}\sup_{\rho_{SA^{m}}}D_H^\epsilon((\mathcal{E}^{\otimes m}(\rho_{SA^{m}}))^{\otimes n}\|(\mathcal{F}^{\otimes m}(\rho_{SA^{m}}))^{\otimes n}) \\ \geq &\frac{1}{nm}\sup_{\rho_{SA^n}}\left\{nD_\alpha(\mathcal{E}^{\otimes m}(\rho_{SA^{m}})\|\mathcal{F}^{\otimes m}(\rho_{SA^{m}}))-\log\left(\frac{1}{\epsilon}\right)+\frac{2}{\alpha-1}\log\left(\frac{1}{1-\sqrt{1-\epsilon}}\right)\right\} \\
     = &\frac{1}{m}D_\alpha(\mathcal{E}^{\otimes m}\|\mathcal{F}^{\otimes m})+\frac{1}{nm}\left(\frac{2}{\alpha-1}\log\left(\frac{1}{1-\sqrt{1-\epsilon}}\right)-\log\left(\frac{1}{\epsilon}\right)\right).
\end{align*}
Here the first inequality follows since the supremum in the second line is over a smaller set of states and the second inequality follows from above. 
Taking the $\liminf_{nm\to\infty}$ on both sides yields 
\begin{align*}
    \liminf_{nm\to\infty}\frac{1}{nm}D_H^\epsilon(\mathcal{E}^{\otimes nm}\|\mathcal{F}^{\otimes nm}) \geq D^\text{reg}_\alpha(\mathcal{E}\|\mathcal{F})
\end{align*} for any $\alpha\in(0,1)$.
Now, since $\lim_{\alpha \uparrow 1}D_\alpha(\rho\|\sigma)=\sup_{\alpha\in(0,1)}D_\alpha(\rho\|\sigma) = D(\rho\|\sigma)$ holds also in infinite dimensions (see e.g. \cite{art:SandwichedRenyInInfDim_RenyDivergencesAsWeightedNonCommutative_Berta_2018}), it follows that
\begin{align*}
    \lim_{\alpha\uparrow1}D_\alpha^\text{reg}(\mathcal{E}\|\mathcal{F})&=\sup_{\alpha\in(0,1)}D_\alpha^\text{reg}(\mathcal{E}\|\mathcal{F})= \sup_{\alpha\in(0,1)}\sup_{m\in\mathbb{N}}\sup_{\rho_{SA^m}}\frac{1}{m}D_\alpha(\mathcal{E}^{\otimes m}(\rho_{SA^m})\|\mathcal{F}^{\otimes m}(\rho_{SA^m})) \\ &= \sup_{m\in\mathbb{N}}\sup_{\rho_{SA^m}}\sup_{\alpha\in(0,1)}\frac{1}{m}D_\alpha(\mathcal{E}^{\otimes m}(\rho_{SA^m})\|\mathcal{F}^{\otimes m}(\rho_{SA^m})) = \sup_{m\in\mathbb{N}}\frac{1}{m}D(\mathcal{E}^{\otimes m}\|\mathcal{F}^{\otimes m}) \\ &= D^\text{reg}(\mathcal{E}\|\mathcal{F}).
\end{align*}
Thus we have the desired result
\begin{align*}
    \liminf_{nm\to\infty}\frac{1}{nm}D_H^\epsilon(\mathcal{E}^{\otimes nm}\|\mathcal{F}^{\otimes nm}) \geq \sup_{\alpha\in(0,1)}D^\text{reg}_\alpha(\mathcal{E}\|\mathcal{F}) = D^\text{reg}(\mathcal{E}\|\mathcal{F}).
\end{align*}
\end{proof}
Therefore quantum Stein's lemma for parallel strategies \eqref{equ:QuantumChannelStein} holds also in infinite dimensions. \\
The statement of the adaptive quantum Stein's lemma for channels is analogous to the parallel one, just employing an adaptive discrimination strategy instead of a parallel one. In the finite dimensional setting it was proved by combining a statement that the asymptotic optimal adaptive scaling is given by the amortized relative channel divergence \cite{wang_resource_2019}, instead the regularized one, and a chain rule for the quantum relative entropy \cite{fang_chain_2020} which implies that the amortized and regularized quantum channel divergences are indeed equivalent. Here we use our one-shot result \autoref{thm:main} to show this equivalence directly and without going through the amortized channel divergence. Note that as mentioned and discussed in \autoref{sec:finiteness_cond}, in infinite dimensions, we are not able to show this equivalence for all channels, but only under the condition that the geometric \renyi divergence between the channels is finite for some $\alpha > 1$. 

\begin{proposition}[Quantum Stein's Lemma for adaptive strategies in infinite dimensions]\label{prop:adaptive_stein}
Let $\mathcal{E},\mathcal{F}:\P(\mathcal{H})\to\P(\mathcal{K})$ be two quantum channels (CPTP maps) where $\mathcal{H}, \mathcal{K}$ are arbitrary separable Hilbert spaces. If there exists $\alpha>1$, s.t. $\Dg(\E\|\F)<\infty$, then
\begin{equation}
    \lim_{\epsilon\to0}\lim_{n\to\infty}\sup_{\rho_1=\sigma_1; \{\Lambda_i\}_i}\frac{1}{n}D_H^\epsilon(\mathcal{E}(\rho_n)\|\mathcal{F}(\sigma_n))= D^{\mathrm{reg}}(\mathcal{E}\|\mathcal{F}).
\end{equation}
\end{proposition}

\begin{proof}
Since adaptive strategies are more general than parallel ones, it follows that
\begin{align}
    \sup_{\rho_1=\sigma_1; \{\Lambda_i\}_i}\frac{1}{n}D_H^\epsilon(\mathcal{E}(\rho_n)\|\mathcal{F}(\sigma_n)) \geq \sup_{\rho} \frac{1}{n}D_H^\epsilon(\mathcal{E}^{\otimes n}(\rho)\|\mathcal{F}^{\otimes n}(\rho)) =  \frac{1}{n}D_H^\epsilon(\mathcal{E}^{\otimes n}\|\mathcal{F}^{\otimes n}).
\end{align}
which implies
\begin{equation}
	D^\text{reg}(\mathcal{E}\|\mathcal{F}) \leq \lim_{\epsilon\to0}\liminf_{n\to\infty}\sup_{\rho_1=\sigma_1; \{\Lambda_i\}_i}\frac{1}{n}D_H^\epsilon(\mathcal{E}(\rho_n)\|\mathcal{F}(\sigma_n))\,.
\end{equation}
If $\widehat{D}_{\alpha}(\mathcal{E}\|\mathcal{F})<\infty$ for some $\alpha>1$, then \autoref{thm:main} holds.
	In \autoref{thm:main}, taking the supremum over parallel input states $\nu_m$, and then the limits $m \to \infty$, $n \to \infty$, $\alpha_a \to 0$, $\alpha_p \to 0$ in this order, we find that:
\begin{equation}
	\lim_{\alpha_a \to 0} \limsup_{n \to \infty} {1 \over n} D_H^{\alpha_a}(\E(\rho_n)\|\F(\sigma_n)) \leq \lim_{\alpha_p \to 0} \limsup_{m \to \infty} \sup_{\nu_m} D_H^{\alpha_p}(\E^{\otimes m}(\nu_m)\|\F^{\otimes m}(\nu_m)) = D^{\mathrm{reg}}(\E\|\F)
\end{equation}
which allows us to conclude that the limit exists without requiring $\liminf$ or $\limsup$, and:
\begin{align*}
    D^\text{reg}(\mathcal{E}\|\mathcal{F}) \leq \lim_{\epsilon\to0}\lim_{n\to\infty}\sup_{\rho_1=\sigma_1; \{\Lambda_i\}_i}\frac{1}{n}D_H^\epsilon(\mathcal{E}(\rho_n)\|\mathcal{F}(\sigma_n)) \leq D^\text{reg}(\mathcal{E}\|\mathcal{F}).
\end{align*}

\end{proof}
These two propositions together prove \autoref{thm:steins}.

\subsection{One-shot relation between adaptive and parallel strategies (\autoref{thm:main})}  \label{sec:ProofofMain}

Our way of proving our one-shot result proceeds by making sure that the main technical lemmas from \cite{art:Bjarne} also hold in the infinite dimensional setting. We will be fairly explicit in our proofs here, ocasionally repeating some of the finite-dimensional arguments, and making use of the formalism of quantum-$f$-divergences and the entropic inequalities we proved in \autoref{sec:InfExtensions}.

\subsubsection{Infinite dimensional one-shot version of the chain rule}
\begin{lemma} (An extended infinite dimensional version of \cite[Lemma 8]{art:Bjarne}, see also \cite[Prop. 3.2]{fang_chain_2020}) 
	\label{lemma:BjarneLemma8}
	\label{lemma:BjarneLemma8Extended}
	Let $\mathcal{E},\mathcal{F}:\P(\mathcal{H})\to\P(\mathcal{K})$ be two quantum channels, where $\mathcal{H},\mathcal{K}$ are separable Hilbert spaces and $\rho,\sigma\in\mathcal{D}(\mathcal{H})$ be some states. Then for any $\epsilon,\epsilon^\prime,\mu >0$ and any $m \in \naturals$ there exists a state $\nu\equiv\nu(m, \mu,\rho,\sigma)\in B^\epsilon(\rho)$ s.t. 
	\begin{align}
		&D_{\max}^{\epsilon+\epsilon^\prime}\big((\mathcal{E}(\rho))^{\otimes m}\|(\mathcal{F}(\sigma))^{\otimes m}\big) \leq D_{\max}^{\epsilon^\prime}\big((\mathcal{E}(\nu))^{\otimes m}\|(\mathcal{F}(\nu)\big)^{\otimes m}\big) + m D_{\max}^{\epsilon/m}(\rho\|\sigma) + \mu \,.
	\end{align}
	Additionally it holds that:
	\begin{equation}
		D_{\max}^{\epsilon+\epsilon^\prime}\big((\mathcal{E}(\rho))^{\otimes m}\|\mathcal{F}(\sigma))^{\otimes m}\big) \leq \sup_{\nu \in \DM} D_{\max}^{\epsilon^\prime}\big((\mathcal{E}(\nu))^{\otimes m}\|(\mathcal{F}(\nu))^{\otimes m}\big) + m D_{\max}^{\epsilon/m}(\rho\|\sigma)\,.
	\end{equation}
\end{lemma}
\begin{proof}
	The proof is essentially the same as for non-infinite dimensional version in \cite{art:Bjarne}, just that we cannot assume that the infimum in the smoothed max-relative entropy is achieved since $B^\epsilon(\rho)$ is a non-compact set, as the underlying Hilbert space is not finite dimensional. \\
	Assume the same requirements on $\mathcal{E},\mathcal{F},\rho,\sigma,\epsilon,\epsilon^\prime,\mu$ as stated in the Lemma.
	Pick $\nu\in B^{\epsilon/m}(\rho)$ s.t. $D^{\epsilon/m}_{\max}(\rho\|\sigma)\geq D_{\max}(\nu\|\sigma)-\frac{\mu}{2}$, i.e. close to 'the optimal smoothing state'. Note that this implies $\nu^{\otimes m} \leq 2^{m D_{\max}^{\epsilon/m}(\rho\|\sigma)} \sigma^{\otimes m}$. 
	Similarly pick $\tau\in B^{\epsilon'}(\mathcal{E}(\nu)^{\otimes m})$ s.t. $D_{\max}(\mathcal{E}(\nu)^{\otimes m}\|\mathcal{F}(\nu)^{\otimes m})\geq  D_{\max}(\tau\|\mathcal{F}(\nu)^{\otimes m})-\frac{\mu}{2}$. \\ 
	Now $\tau\leq2^{D_{\max}(\tau\|\mathcal{F}(\nu))^{\otimes m}}\mathcal{F}(\nu)^{\otimes m}\leq 2^{D^{\epsilon'}_{\max}(\mathcal{E}(\nu)^{\otimes m}\|\mathcal{F}(\nu)^{\otimes m})+\frac{\mu}{2}}2^{m D^{\epsilon/m}_{\max}(\rho\|\sigma)+\frac{\mu}{2}}\mathcal{F}(\sigma)$, where we used that $\omega\geq\tilde{\omega}\implies\mathcal{E}(\omega)\geq\mathcal{E}(\tilde{\omega})$ for any two states $\omega,\tilde{\omega}$.
	Now by the triangle inequality and DPI for the sine distance we have $P(\tau,\mathcal{E}(\rho)^{\otimes m})\leq P(\tau,\mathcal{E}(\nu)^{\otimes m})+P(\mathcal{E}(\nu)^{\otimes m},\mathcal{E}(\rho)^{\otimes m})\leq \epsilon^\prime+P(\nu^{\otimes m},\rho^{\otimes m})\leq \epsilon^\prime+m \epsilon/m = \epsilon' + \epsilon$.
	
	The second inequality follows by taking the supremum over $\nu$ on the right-hand side and then the limit $\mu \to 0$. 
\end{proof}

\subsubsection{Continuity of the Petz-\renyi divergence in \texorpdfstring{$\alpha$}{alpha}}
First we establish a continuity bound on the $\alpha$-Petz-\renyi divergence in $\alpha$ (\autoref{lemma:Bjarne9}), after which we establish that the explicit bound on the convergence speed of the asymptotic equipartition property established in \cite[Lemma 11]{art:Bjarne} holds equally in infinite dimensions by way of the two inequalities \eqref{equ:ImpIneq2NEW}, \eqref{equ:ImpINeq3REP}. 

\begin{lemma}[infinite dimensional version of Lemma 9 in \cite{art:Bjarne}, simplified]
\label{lemma:Bjarne9}
Let $\rho,\sigma\in\mathcal{D}(\mathcal{H})$ be any two states on a separable Hilbert space $\mathcal{H}$ and $\gamma\in(0,1]$. \\ Define
$c_\gamma(\rho\|\sigma):=\frac{1}{\gamma}\log\left(2^{\gamma D_{1+\gamma}(\rho\|\sigma)}+2^{-\gamma D_{1-\gamma}(\rho\|\sigma)}+1 \right)$. Then for all $\gamma\in(0,1]$ and $0<\delta\leq\frac{\gamma}{2}$:
\begin{equation}
    D_{1+\delta}(\rho\|\sigma)\leq D(\rho\|\sigma)+\delta(c_\gamma(\rho\|\sigma))^2.
    \label{equ:PetzRenyiUpperContinuity}
\end{equation}
Furthermore, if $D(\rho\|\sigma)<\infty$, then for all $\gamma\in(0,1]$ and $0<\delta\leq\frac{\log3}{2c_\gamma(\rho\|\sigma)}$
\begin{equation}
    D_{1-\delta}(\rho\|\sigma)\geq D(\rho\|\sigma)-\delta(c_\gamma(\rho\|\sigma))^2.
     \label{equ:PetzRenyiLowerContinuity}
\end{equation}
\end{lemma}

\begin{proof}
Recall the definitions of the Umegaki relative entropy and the Petz-\renyi relative entropies via the relative modular operator as 
\begin{align*}
    &D(\rho\|\sigma) = \int_{(0,\infty)}\log\lambda\langle\sqrt{\rho}|\xi^{\Delta_{\rho,\sigma}}(d\lambda)|\sqrt{\rho}\rangle = \langle\sqrt{\rho}|\log\Delta_{\rho,\sigma}|\sqrt{\rho}\rangle, \\
    &D_\alpha(\rho\|\sigma) = \frac{1}{\alpha-1}\log\int_{(0,\infty)}\lambda^{\alpha-1}\langle\sqrt{\rho}|\xi^{\Delta_{\rho,\sigma}}(d\lambda)|\sqrt{\rho}\rangle = \frac{1}{\alpha-1}\log\langle\sqrt{\rho}|\Delta_{\rho,\sigma}^{\alpha-1}|\sqrt{\rho}\rangle,  \\
    &D_{1+\delta}(\rho\|\sigma) = \frac{1}{\delta}\log\int_{(0,\infty)}\lambda^{\delta}\langle\sqrt{\rho}|\xi^{\Delta_{\rho,\sigma}}(d\lambda)|\sqrt{\rho}\rangle = \frac{1}{\delta}\log\langle\sqrt{\rho}|\Delta_{\rho,\sigma}^{\delta}|\sqrt{\rho}\rangle.
\end{align*}
To prove the lemma we follow the proof in \cite[Proof of Lemma 9]{art:Bjarne}, essentially replacing $X=\rho\otimes(\sigma^{-1})^T$ by $\Delta_{\rho,\sigma}$ and $|\phi\rangle=\sum_i\sqrt{\rho}|i\rangle\otimes|i\rangle\in\mathcal{H}\otimes\mathcal{H}$ with $|\sqrt{\rho}\rangle\in\mathcal{B}_2(\mathcal{H})$. It is repeated here  for the convenience of the reader.
\\ We write $t^\delta=1+\delta\ln(t)+r_\delta(t)$, where the first two summands are the Taylor-coefficients of $t^\delta$ when expanding in $\delta$ around $\delta=0$. Set $r_\delta(t):=t^\delta-\delta\ln(t)-1$. Using $e^x\geq 1+x$ we can upper bound it by
\begin{align*}
    r_\delta(t) &= t^\delta-\delta\ln(t)-1 \leq t^\delta+e^{-\delta\ln(t)}-2 = e^{\delta\ln(t)}+e^{-\delta\ln(t)}-2 \\
&=  2(\cosh(\delta\ln(t))-1) =: s_\delta(t).
\end{align*}
It is easy to check that $s_\delta(t)$ is monotonically increasing in $t$ on $[1,\infty)$, concave in $t$ on $[3,\infty)$ if $\delta\leq\frac{1}{2}$, and satisfies $s_{-\delta}(t)=s_{\delta}(t) =s_{\gamma\delta}(t^{\frac{1}{\gamma}})$. We get for $t\geq0$ and $\gamma\in(0,1]$
\begin{align*}
    s_\delta(t)= s_{\frac{\delta}{\gamma}}(t^\gamma)\leq s_{\frac{\delta}{\gamma}}(t^\gamma+t^{-\gamma}) \leq  s_{\frac{\delta}{\gamma}}(t^\gamma+t^{-\gamma}+1).
\end{align*}
It is easy to see that $t^\gamma+t^{-\gamma}+1\geq3$ for all $t \in (0,\infty)$. Thus, we can use Jensens inequality to get
\begin{align*}
    \langle\sqrt{\rho}|s_\delta(\Delta_{\rho,\sigma})|\sqrt{\rho}\rangle &= \int_{(0,\infty)}s_\delta(\lambda)d\mu_\rho(\lambda) \leq \int_{(0,\infty)}s_{\frac{\delta}{\gamma}}(\lambda^\gamma+\lambda^{-\gamma}+1)d\mu_\rho(\lambda) \\ &\overset{\text{Jensen}}{\leq} s_{\frac{\delta}{\gamma}}\left(\int_{(0,\infty)}(\lambda^\gamma+\lambda^{-\gamma}+1)d\mu_\rho(\lambda)\right) = s_{\frac{\delta}{\gamma}}\left(2^{\gamma c_\gamma(\rho\|\sigma)}\right).
\end{align*}
where we wrote d$\mu_\rho(\lambda)\equiv\Tr[\rho \text{d}\xi^{\Delta_{\rho,\sigma}}(\lambda)]= \langle\sqrt{\rho}|\text{d}\xi^{\Delta_{\rho,\sigma}}(\lambda)|\sqrt{\rho}\rangle$, where d$\xi^{\Delta_{\rho,\sigma}}(\lambda)$ is the spectral measure of the relative modular operator $\Delta_{\rho,\sigma}$.
Now using Taylor`s theorem with the Lagrange remainder we can bound
\begin{align*}
    s_\delta(t)&=s_0(t)+\frac{d}{d\delta}s_\delta(t)|_{\delta=0}\delta+\frac{1}{2}\frac{d^2}{d\delta^2}s_\delta(t)|_{\delta=\xi}\delta^2 \\ &=\delta^2(\ln(t))^2\cosh(\xi\ln(t))\leq\delta^2(\ln(t))^2\cosh(\delta\ln(t)) 
\end{align*}
for all $t\geq0$ and some $\xi\in(0,\delta)$, where we used that $s_0(t)=\frac{d}{d\delta}s_\delta(t)|_{\delta=0}=0$. Hence
\begin{align*}
    \langle\sqrt{\rho}|s_\delta(\Delta_{\rho,\sigma})|\sqrt{\rho}\rangle \leq s_{\frac{\delta}{\gamma}}\left(2^{\gamma c_\gamma(\rho\|\sigma)}\right) \leq (\delta\ln(2)c_\gamma(\rho\|\sigma))^2\cosh(\delta\ln(2)c_\gamma(\rho\|\sigma)).
\end{align*}
We can now apply this to derive the upper bound in the Lemma. For $\delta >0$ have
\begin{align*}
    D_{1+\delta}(\rho\|\sigma) &= \frac{1}{\delta}\log\langle\sqrt{\rho}|\Delta_{\rho,\sigma}^{\delta}|\sqrt{\rho}\rangle = \frac{1}{\delta}\log\langle\sqrt{\rho}|1+\delta\ln(2)\log\Delta_{\rho,\sigma}+r_\delta(\Delta_{\rho,\sigma})|\sqrt{\rho}\rangle \\
    &= \frac{1}{\delta}\log(1+\delta\ln(2)D(\rho\|\sigma)+\langle\sqrt{\rho}|r_\delta(\Delta_{\rho,\sigma})|\sqrt{\rho}\rangle) \\
    &\leq \frac{1}{\delta}\log(1+\delta\ln(2)D(\rho\|\sigma)+\langle\sqrt{\rho}|s_\delta(\Delta_{\rho,\sigma})|\sqrt{\rho}\rangle) \\
    &= \frac{1}{\delta}\log(1+\delta\ln(2)D(\rho\|\sigma)) +  \frac{1}{\delta}\log\left(1+\frac{\langle\sqrt{\rho}|s_\delta(\Delta_{\rho,\sigma})|\sqrt{\rho}\rangle}{1+\delta\ln(2)D(\rho\|\sigma)}\right) \\
    &\leq D(\rho\|\sigma) + \frac{1}{\delta}\log\left(1+\langle\sqrt{\rho}|s_\delta(\Delta_{\rho,\sigma})|\sqrt{\rho}\rangle\right) \\
    &\leq D(\rho\|\sigma) + \frac{1}{\delta}\log\left(1+(\delta\ln(2)c_\gamma(\rho\|\sigma))^2\cosh(\delta\ln(2)c_\gamma(\rho\|\sigma))\right) \\
    &\leq D(\rho\|\sigma) +\delta\ln(2)(c_\gamma(\rho\|\sigma))^2,
\end{align*}
where in the third line it was used that $\log(1+x)\leq\frac{x}{\ln(2)}$ and $\delta>0$. The final inequality follows from the fact that $k\mapsto k^2-\ln(1+k^2\cosh(k))$ is monotonically increasing and hence positive.
For the lower bound in the Lemma, i.e. the case $\delta<0$, a slightly different argument has to be applied.
\begin{align*}
    D_{1+\delta}(\rho\|\sigma) &= \frac{1}{\delta}\log\langle\sqrt{\rho}|\Delta_{\rho,\sigma}^{\delta}|\sqrt{\rho}\rangle = \frac{1}{\delta}\log\langle\sqrt{\rho}|1+\delta\ln(2)\log\Delta_{\rho,\sigma}+r_\delta(\Delta_{\rho,\sigma})|\sqrt{\rho}\rangle \\
    &\geq \frac{1}{\delta}\log(1+\delta\ln(2)D(\rho\|\sigma)+\langle\sqrt{\rho}|s_\delta(\Delta_{\rho,\sigma})|\sqrt{\rho}\rangle) \\
    &\geq D(\rho\|\sigma) + \frac{1}{\delta\ln(2)}\langle\sqrt{\rho}|s_\delta(\Delta_{\rho,\sigma})|\sqrt{\rho}\rangle \\ &\geq D(\rho\|\sigma) +\delta\ln(2)(c_\gamma(\rho\|\sigma))^2\cosh(\delta\ln(2)c_\gamma(\rho\|\sigma)),
\end{align*}
where in the second inequality again $\log(1+x)\leq\frac{x}{\ln(2)}$ was used. If now $|\delta|\leq\frac{\log(3)}{2c_\gamma(\rho\|\sigma)}\geq\frac{\gamma}{2}$, then $\ln(2)\cosh(\ln(3)/2)<1$ and thus
\begin{align*}
    D_{1+\delta}(\rho\|\sigma) \geq D(\rho\|\sigma) + \delta(c_\gamma(\rho\|\sigma))^2.
\end{align*}
\end{proof}

\subsubsection{Asymptotic equipartition property}
With \autoref{lemma:Bjarne9} and the inequalities \eqref{equ:ImpIneq2}, \eqref{equ:ImpIneq3}, the authors of \cite{art:Bjarne} derive the following non-asymptotic $n$ shot version of the asymptotic equipartition property (AEP) of the smoothed-max-relative entropy. Having established these in the infinite dimensional setting allows us to conclude the analogous result:

\begin{lemma}[Infinite dimensional version of Lemma 11 in \cite{art:Bjarne}, simplified] \label{lem:InfDImDmaxAEP}
Let $\rho,\sigma\in\mathcal{D}(\mathcal{H})$ be any states on a separable Hilbert space $\mathcal{H}$, and for $\gamma\in(0,1]$ define $c_\gamma(\rho\|\sigma)$ as in \autoref{lemma:Bjarne9}. Then for all $\epsilon\in(0,1]$ and $n\in\mathbb{N}$
\begin{align}
    &\frac{1}{n}D^\epsilon_{\max}(\rho^{\otimes n}\|\sigma^{\otimes n})\geq D(\rho\|\sigma)-\frac{4c_\gamma(\rho\|\sigma)}{\sqrt{n}}\log\left(\frac{2}{1-\epsilon}\right), \label{equ:DMaxAEPlowerbound} \\
    &\frac{1}{n}D^\epsilon_{\max}(\rho^{\otimes n}\|\sigma^{\otimes n})\leq D(\rho\|\sigma)+\frac{4c_\gamma(\rho\|\sigma)}{\sqrt{n}}\log\left(\frac{2}{\epsilon}\right)+\frac{1}{n}\log\left(\frac{1}{1-\epsilon^2}\right). \label{equ:DMaxAEPupperbound}
\end{align}
\end{lemma}
\begin{proof}
The proof may be found in \cite{art:Bjarne}. Since it is based on inequalities \eqref{equ:ImpIneq2NEW}, \eqref{equ:ImpINeq3REP} and \autoref{lemma:Bjarne9}, all of which we established to hold in infinite dimensions, it also holds. For the reader's convenience the argument is repeated in \autoref{app:ProofOfDmaxAEPLemma}. 
\end{proof}

\begin{remark}
	The authors of \cite{fawzi_asymptotic_2023} also recently proved a different version of an AEP for the smoothed max divergence in infinite dimensions. However, their version is not suitable for our applications. Specifically, they show what is known as second-order asymptotics, where the second order term is controlled by the relative entropy variance $V(\rho\|\sigma)$. For our result we subsequently need to upper bound the error terms in the AEP by something that satisfies a chain rule (we will use the geometric \renyi divergence for this purpose), and we are not aware of any way to do this for the relative entropy variance, which is why we cannot make use of the results from \cite{fawzi_asymptotic_2023}. For a similar discussion in finite dimensions, see also \cite[Remark 13]{art:Bjarne}.
\end{remark}

\subsubsection{Proof of \autoref{thm:main}}

Having established these above lemmas we can proceed to give the proof of the main theorem. It now proceeds ba a very similar argument as the proof in \cite[Theorem 6]{art:Bjarne}. 
\begin{proof}[Proof of \autoref{thm:main}]
Let an adaptive channel discrimination strategy between $\mathcal{E}$ and $\mathcal{F}$ through $n$ channels uses be given. Then, inequality \eqref{equ:ImpIneq4} implies
\begin{align*}
    D_H^{\alpha_a}(\mathcal{E}(\rho_n)\|\mathcal{F}(\sigma_n))\leq \frac{1}{n}\frac{1}{1-\alpha_a}[D(\mathcal{E}(\rho_n)\|\mathcal{F}(\sigma_n))+h(\alpha_a)].
\end{align*}
Now we rewrite 
\begin{align*}
    D(\mathcal{E}(\rho_n)\|\mathcal{F}(\sigma_n)) &= D(\mathcal{E}(\rho_n)\|\mathcal{F}(\sigma_n)) - D(\rho_n\|\sigma_n)+D(\Lambda_n(\mathcal{E}(\rho_{n-1}))\|\Lambda_n(\mathcal{F}(\sigma_{n-1}))) \\ &\overset{\text{DPI}}{\leq} D(\mathcal{E}(\rho_n)\|\mathcal{F}(\sigma_n)) - D(\rho_n\|\sigma_n) + D(\mathcal{E}(\rho_{n-1})\|\mathcal{F}(\sigma_{n-1})) \\ & \overset{\text{iterate}}{\leq ... \leq} \sum_{k=1}^n[D(\mathcal{E}(\rho_k)\|\mathcal{F}(\sigma_k))-D(\rho_k\|\sigma_k)] \\
    &\leq n(D(\mathcal{E}(\rho_l)\|\mathcal{F}(\sigma_l))-D(\rho_l\|\sigma_l)),
\end{align*} where $l:=\argmax_{k\in\{1,...,n\}}[D(\mathcal{E}(\rho_k)\|\mathcal{F}(\sigma_k))-D(\rho_k\|\sigma_k)]$.
Then with \autoref{lem:InfDImDmaxAEP} we can upper bound this relative entropy difference by a smoothed max-relative entropy difference and an error term. 
\begin{align*}
    D(\mathcal{E}(\rho_l)\|\mathcal{F}(\sigma_l))-D(\rho_l\|\sigma_l) &\leq \frac{1}{m}\left[D^{\epsilon_1}_{\max}(\mathcal{E}(\rho_l)^{\otimes m}\|\mathcal{F}(\sigma_l)^{\otimes m})-D_{\max}^{\epsilon_2}(\rho_l^{\otimes m}\|\sigma_l^{\otimes m})\right] \\ &+\frac{1}{\sqrt{m}}\left[4c_{\gamma_1}(\mathcal{E}(\rho_l)\|\mathcal{F}(\sigma_l))\log\left(\frac{2}{1-\epsilon_1}\right)+4c_{\gamma_2}(\rho_l\|\sigma_l)\log\frac{2}{\epsilon_2}\right] + \frac{1}{m}\log\left(\frac{1}{1-\epsilon_2^2}\right)
\end{align*} for any $\gamma_{1,2}\in(0,1]$ and $\epsilon_{1,2}\in(0,1]$. Setting $c_l:=4\inf_{\gamma_1,\gamma_2}[c_{\gamma_1}(\mathcal{E}(\rho_l)\|\mathcal{F}(\sigma_l))+c_{\gamma_2}(\rho_l\|\sigma_l)]$ and $\epsilon=:\epsilon_2=1-\epsilon_1$ gives
\begin{align*}
    D(\mathcal{E}(\rho_l)\|\mathcal{F}(\sigma_l))-D(\rho_l\|\sigma_l) &\leq \frac{1}{m}\left[D^{1-\epsilon}_{\max}(\mathcal{E}(\rho_l)^{\otimes m}\|\mathcal{F}(\sigma_l)^{\otimes m})-D_{\max}^{\epsilon}(\rho_l^{\otimes m}\|\sigma_l^{\otimes m})\right] +\frac{c_l}{\sqrt{m}}\log\frac{2}{\epsilon}+\frac{1}{m}\log\frac{1}{1-\epsilon^2}.
\end{align*}
Now, due to the fact that the Petz-Renyi divergence is positive when evaluated on normalized states, and upper bounded by the geometric divergence (as in \eqref{eq:relation_renyi_entropies}), we get 
\begin{equation}
	    c_\gamma(\E(\rho_l)\|\F(\sigma_l)) 
	    \leq \frac{1}{\gamma}\log( 2^{\gamma \widehat{D}_{1 + \gamma} (\E(\rho_l)\|\F(\sigma_l))} + 2)\,.
	\end{equation}
	Now, by repeated use of the chain rule for the geometric \renyi divergence (\autoref{thm:geometric_chain_rule}), we get
	\begin{align}
		\widehat{D}_{1 + \gamma}(\E(\rho_l)\|\F(\sigma_l)) &\leq \widehat{D}_{1 + \gamma}(\E\|\F) + \widehat{D}_{1 + \gamma}(\rho_l\|\sigma_l)\\
		&= \widehat{D}_{1 + \gamma}(\E\|\F) + \widehat{D}_{1 + \gamma}(\Lambda_l(\E(\rho_{l-1}))\|\Lambda_l(\F(\sigma_{l-1}))) \\
		&\leq \widehat{D}_{1 + \gamma}(\E\|\F) + \widehat{D}_{1 + \gamma}(\E(\rho_{l-1})\|\F(\sigma_{l-1})) \\
		&\leq \ldots \leq l \widehat{D}_{1 + \gamma}(\E\|\F)\,.
	\end{align}
	Defining the corresponding channel quantity
	\begin{equation}
	    \widehat{c}_\gamma(\E\|\F) \coloneqq \frac{1}{\gamma}\left(2^{\gamma \widehat{D}_{1 + \gamma}(\E\|\F)} + 2\right) \, ,
	\end{equation}
	we find that 
	\begin{equation}
	    c_\gamma(\E(\rho_l)\|\F(\sigma_l)) \leq \frac{1}{\gamma} \log(2^{l \widehat{D}_{1 + \gamma}(\E\|\F)} + 2) \leq \frac{l}{\gamma} \log(2^{\widehat{D}_{1 + \gamma}(\E\|\F)} + 2) = l\, \widehat{c}_\gamma(\E\|\F) \,.
	\end{equation}
    Using the same argument we find that also $\widehat{D}_{1 + \gamma}(\rho_l \|\sigma_l) \leq l \widehat{D}_{1 + \gamma}(\E\|\F)$ and hence also
	\begin{equation}
	    c_\gamma(\rho_l\|\sigma_l) \leq l\, \widehat{c}_\gamma(\E\|\F)\,.
	\end{equation}
	
	Thus,
	\begin{equation}
	    c_l \leq 8 l \inf_{\gamma \in (0,1]} \widehat{c}_\gamma(\E\|\F) \leq 8 n \inf_{\gamma \in (0,1]} \widehat{c}_\gamma(\E\|\F)\,,
	\end{equation}
	Remember that we assumed that, $\Dg(\mathcal{E}\|\mathcal{F})<\infty$ for some $\alpha > 1$, and hence (as the geometric \renyi divergence is increasing in $\alpha$, which is easy to see from the definition through minimal reverse tests) it follows that $c_l$ is upper bounded as $c_l\leq C\cdot l\leq C\cdot n$, where $C=8 \inf_{\alpha \in (1, 2]} {1 \over {\alpha - 1}}\log\left(2^{(\alpha - 1) D_{\alpha}(\mathcal{E}\|\mathcal{F})}+2\right)$.

Now, applying \autoref{lemma:BjarneLemma8} with some $\mu>0$ to $D^{1-\epsilon}_{\max}(\mathcal{E}(\rho_l)^{\otimes m}\|\mathcal{F}(\sigma_l)^{\otimes m})$ we get a state $\nu(\epsilon,\mu,l,m;\rho)\equiv\nu\in B^\epsilon(\rho_l^{\otimes m})$, s.t.
\begin{align*}
    D_{\max}^{1-2\epsilon+\epsilon}(\mathcal{E}(\rho_l)^{\otimes m}\|\mathcal{F}(\sigma_l)^{\otimes m})\leq D_{\max}^{1-2\epsilon}(\mathcal{E}^{\otimes m}(\nu)\|\mathcal{F}^{\otimes m}(\nu)) + D_{\max}^\epsilon(\rho_l^{\otimes m}\|\sigma_l^{\otimes m})+ \mu.
\end{align*}
So
\begin{align*}
    D(\mathcal{E}(\rho_l)\|\mathcal{F}(\sigma_l))-D(\rho_l\|\sigma_l) \leq \frac{1}{m}\left[D_{\max}^{1-2\epsilon}(\mathcal{E}^{\otimes m}(\nu)\|\mathcal{F}^{\otimes m}(\nu))+\mu\right]+\frac{c_l}{\sqrt{m}}\log\frac{2}{\epsilon}+\frac{1}{m}\log\frac{1}{1-\epsilon^2}.
\end{align*}
Now, to convert the max-divergence on the r.h.s of this expression to a hypothesis testing divergence we apply inequality \eqref{equ:ImpIneq2NEW} to this and get
\begin{align*}
   D(\mathcal{E}(\rho_l)\|\mathcal{F}(\sigma_l))-D(\rho_l\|\sigma_l) \leq \frac{1}{m}D_H^{1-(1-2\epsilon)^2}(\mathcal{E}^{\otimes m}(\nu)\|\mathcal{F}^{\otimes m}(\nu))+\frac{c_l}{\sqrt{m}}\log\frac{2}{\epsilon}+\frac{1}{m}\left[\log\frac{1}{1-\epsilon^2}+\log\frac{1}{1-(1-2\epsilon)^2}+\mu\right].
\end{align*} Setting $\alpha_p:=1-(1-2\epsilon)^2\Leftrightarrow \epsilon=\frac{1}{2}(1-\sqrt{1-\alpha_p})$ we get
\begin{align*}
    D(\mathcal{E}(\rho_l)\|\mathcal{F}(\sigma_l))-D(\rho_l\|\sigma_l) \leq \frac{1}{m}D^{\alpha_p}_H(\mathcal{E}^{\otimes m}(\nu)\|\mathcal{F}^{\otimes m}(\nu))+\frac{c_l}{\sqrt{m}}\log\frac{8}{\alpha_p}+\frac{1}{m}\left[\log\frac{1}{\alpha_p}-\log(1-\frac{\alpha_p}{4})+\mu\right]
\end{align*}
where we used, that $2\epsilon=1-\sqrt{1-\alpha_p}\geq \frac{\alpha_p}{2}$ and $1-\epsilon^2\geq 1-\frac{\alpha_p}{2}$.
So now putting everything together we get
\begin{align*}
    \frac{1}{n}D^{\alpha_a}_H(\mathcal{E}(\rho_n)\|\mathcal{F}(\sigma_n))\leq \frac{1}{1-\alpha_a}\left[\frac{1}{m}D^{\alpha_p}_H(\mathcal{E}^{\otimes m}(\nu)\|\mathcal{F}^{\otimes m}(\nu))+\frac{c_l}{\sqrt{m}}\log\frac{8}{\alpha_p}+\frac{1}{m}\left[\log\frac{1}{\alpha_p}-\log(1-\frac{\alpha_p}{4})+\mu\right]+\frac{h(\alpha_a)}{n}\right],
\end{align*} where $\mu$ can be chosen as small as desired.
\end{proof}

\subsection{Proof of \autoref{thm:InfDimChainRule} (Chain rule for the quantum relative entropy)}
The chain rule (\autoref{thm:InfDimChainRule}) for the quantum relative entropy in infinite dimensions, can be be proved quite straightforwardly combining \autoref{lem:InfDImDmaxAEP} and \autoref{lemma:BjarneLemma8}:
\begin{proof}
	
	\newcommand{\m}{^{\otimes m}}
	\newcommand{\nm}{^{\otimes nm}}
	Applying \autoref{lemma:BjarneLemma8} with $\rho \leftarrow \rho^{\otimes n}$, $\sigma \leftarrow \sigma^{\otimes n}$, $\E \leftarrow \E^{\otimes n}$ and $\F \leftarrow \F^{\otimes n}$ we get
	\begin{equation}
		D_{\max}^{\epsilon+\epsilon^\prime}\big((\mathcal{E}(\rho))\nm\|(\mathcal{F}(\sigma))\nm\big) \leq \sup_{\nu \in \DM[\HS\n]} D_{\max}^{\epsilon^\prime}\big((\mathcal{E}\n(\nu))\m\|(\mathcal{F}\n(\nu))\m\big) + m D_{\max}^{\epsilon/m}(\rho\n\|\sigma\n)\,.
	\end{equation}
	Now, dividing my $nm$ and using \autoref{lem:InfDImDmaxAEP} we get 
	\begin{multline}
		D(\E(\rho)\|\F(\sigma)) - \mathcal{O}\br{1 \over \sqrt{n m}} \\
		\leq \frac{1}{n} \sup_{\nu \in \DM[\HS\n]} \br{D(\E\n(\nu)\|\F\n(\nu)) + {4 c_{\gamma}(\E\n(\nu)\|\F\n(\nu))  \over \sqrt{m}}\log\br{2 \over \epsilon} + \mathcal{O}\br{1 \over nm}} \\
		+ D(\rho\|\sigma) + \mathcal{O}\br{\log\br{2m \over \epsilon} \frac{1}{\sqrt{n}}} + \mathcal{O}\br{1 \over n}
	\end{multline}
	Similarly to \cite{art:Bjarne}, we now have that
	\begin{equation}
		c_{\gamma}(\E^{\otimes n}(\nu)\|\F^{\otimes n}(\nu)) \leq \widehat{c}_{\gamma}(\E\n\|\F\n) = n \widehat{c}_{\gamma}(\E\|\F) \leq n \widehat{c}(\E\|\F)
	\end{equation}
	where
	\begin{equation}
		\widehat{c}_{\gamma}(\E\|\F) \coloneqq {1 \over \gamma} \log(2^{\gamma \widehat{D}_{1 + \gamma}(\E\|\F)} + 2) \qquad
		\widehat{c}(\E\|\F) \coloneqq \inf_{\gamma \in (0,1]} \widehat{c}_{\gamma}(\E\|\F)\,.
	\end{equation}
	and the equality in the first line follows from the chain rule for the geometric \renyi divergence (\autoref{thm:geometric_chain_rule}), which implies additivity of the channel divergence.
	Now, by assumption, $\widehat{c}(\E\|\F)$ is finite. Hence, we can take $m = n^3$ and then in the limit $n \to \infty$ all error terms will disappear, so that we get the desired expression.  
\end{proof}

\section{Outlook}
In this final section we discuss multiple ways in which our results could possible be extended.

Recall that we stated our main theorems under the condition that our two channels satisfy $\Dg(\E||\F) <\infty$ for some $\alpha > 1$. As discussed in \autoref{sec:finiteness_cond} this condition is not required when restricting to classical channels or finite dimensional quantum channels. We would hence conjecture that it is unnecessary also in the infinite dimensional quantum case, however we do not currently know how to prove this. In addition, it would be interesting to obtain simple characterizations of channels for which $\Dg(\E\|\F) < \infty$ holds for some $\alpha > 1$. In \autoref{sec:finiteness_cond} we establish a simple condition for generalized depolarizing channels, but we leave finding such conditions for more general channels open to further work.

Since we extend the finite dimensional results of quantum channel discrimination to separable Hilbert spaces, it would seem a natural step to attempt an extension of our results to the more general setting of von Neumann algebras, where quantum hypothesis testing and state discrimination have also already been studied \cite{art:cmp/1104248844,fawzi_asymptotic_2023,art:SandwichedRenyInInfDim_RenyDivergencesAsWeightedNonCommutative_Berta_2018, art:Jen_ov__2018_1,art:Jen_ov__2021_2}.

While we proved quantum Stein's lemma for channels in \autoref{thm:steins}, we were only able to prove the so-called weak converse, where we show that the achievable rate is optimal for all rates that have the type I error also going to zero. In many cases, also a so-called strong converse holds, which states that allowing the type I error to be at any value (just strictly less than one) does not allow for a better asymptotic error decay rate of the type II error. This holds for example in the case of simple asymmetric quantum state discrimination in finite dimensions. Given that this problem is still open though, both in the case of finite dimensional channel discrimination (see e.g. \cite{fang_towards_2022}) and in the case of infinite dimensional state discrimination \cite{mosonyi_strong_2022} we do not tackle this problem here.

\addcontentsline{toc}{section}{Bibliography}
\printbibliography

\appendix
\section{Proofs of elementary properties of the max-relative and Hypothesis testing divergence}

\subsection{Properties of the max-relative entropy} \label{app:MaxRel}
Recall the definition of the max-relative entropy between two states $\rho,\sigma\in\mathcal{D}(\mathcal{H})$ on a separable Hilbert space as
\begin{align*}
    D_\text{max}(\rho||\sigma) :=\inf\{\lambda\in\mathbb{R}|\,\rho\leq 2^\lambda\sigma\}.    
\end{align*}
To prove the data-processing inequality, note that for any positive (and thus of course also for any completely positive), trace-preserving linear map $\Lambda:\P(\mathcal{H})\to\P(\mathcal{K})$ it holds that if $\rho\leq 2^\lambda\sigma\Leftrightarrow 2^\lambda\sigma-\rho\geq0$ then $\Lambda(2^\lambda\sigma-\rho)\geq0 \Leftrightarrow 2^\lambda\Lambda(\sigma)\geq\Lambda(\rho)$. Now taking the infimum over all $\lambda$ s.t. $\rho\leq2^\lambda\sigma$ gives the DPI
\begin{align*}
    D_{\text{max}}(\Lambda(\rho)||\Lambda(\sigma)) \leq     D_{\text{max}}(\rho||\sigma).
\end{align*}
For the triangle inequality, let $\rho,\sigma,\omega\in\mathcal{D}(\mathcal{H})$. If $\rho\leq2^\mu\omega$ and $\omega\leq2^\eta\sigma$, then $\rho\leq { 2^{\mu + \eta}} \sigma$. 
Taking infima over suitable $\mu, \eta$ gives the desired result
\begin{align*}
    D_\text{max}(\rho||\omega)+D_{\text{max}}(\omega||\sigma)\geq D_\text{max}(\rho||\sigma).
\end{align*}
By essentially the same argument we also have
\begin{align*}
    D_{\text{max}}(\rho\otimes\tau||\sigma\otimes\omega) \leq D_{\text{max}}(\rho||\sigma)+D_{\text{max}}(\tau||\omega).
\end{align*}

\subsection{Data processing inequality (DPI) for the Hypothesis testing divergence} \label{app:HTD}
Recall the definition of the hypothesis testing relative entropy between two states $\rho,\sigma\in\mathcal{D}(\mathcal{H})$ on a separable Hilbert space as
\begin{align*}
D_H^\epsilon(\rho||\sigma):=-\log\inf_{0\leq F\leq\1}\{\Tr[F\sigma]|\Tr[F\rho]\geq 1-\epsilon\}.
\end{align*}
Now to prove the data-processing inequality use that if $\Lambda:\P(\mathcal{H})\to\P(\mathcal{K})$ is a completely positive trace preserving map, then its adjoint w.r.t the Hilbert Schmidt inner product, $\Lambda^*$, is a completely positive unital map. Thus
\begin{align*}
    0\leq F\leq \1 \implies 0\leq \Lambda^*(F)\leq \1,
\end{align*} where the same obviously holds for $\1-F$. Thus
\begin{align*}
    D_H^\epsilon(\Lambda(\rho)||\Lambda(\sigma))&=-\log\inf_{0\leq F\leq\1}\{\Tr[F\Lambda(\sigma)]|\Tr[F\Lambda(\rho)]\geq 1-\epsilon\} \\ &= -\log\inf_{0\leq F\leq\1}\{\Tr[\Lambda^*(F)\sigma]|\Tr[\Lambda^*(F)\rho]\geq 1-\epsilon\} \\ &\leq -\log\inf_{0\leq E\leq\1}\{\Tr[E\sigma]|\Tr[E\rho]\geq 1-\epsilon\} = D_H^\epsilon(\rho||\sigma).
\end{align*}

\section{Proof of upper bound in parallel channel quantum Stein's lemma (\autoref{thm:steins})} \label{app:CQSLUpperBound}
\begin{lemma}
Let $\mathcal{E},\mathcal{F}:\P(\mathcal{H}_A)\to\P(\mathcal{K})$ be two quantum channels (CPTP maps) where $\mathcal{H}, \mathcal{K}$ are arbitrary separable Hilbert spaces. Then
\begin{equation}
    \lim_{\epsilon\to0}\limsup_{n\to\infty}\frac{1}{n}D_H^\epsilon(\mathcal{E}^{\otimes n}||\mathcal{F}^{\otimes n})\leq D^{\mathrm{reg}}(\mathcal{E}||\mathcal{F}).
\end{equation}
\end{lemma}
\begin{proof} \cite{wang_resource_states_2019}
We use inequality \eqref{equ:ImpIneq4} on the states $\mathcal{E}^{\otimes n}(\omega_{RA^n})$ and $\mathcal{F}^{\otimes n}(\omega_{RA^n})$ and optimize over $\omega_{RA^n}\in\mathcal{D}(\mathcal{H}_R\otimes\mathcal{H}_A^{\otimes n})$ to get
\begin{align*}
    D_H^\epsilon(\mathcal{E}^{\otimes n}||\mathcal{F}^{\otimes n}) \leq \frac{1}{1-\epsilon}[D(\mathcal{E}^{\otimes n}||\mathcal{F}^{\otimes n})+h(\epsilon)].
\end{align*}
Now multiplying both sides of the above inequality by $\frac{1}{n}$ and taking the limit $n\to\infty$ yields
\begin{align*}
    &\limsup_{n\to\infty}\frac{1}{n}D_H^\epsilon(\mathcal{E}^{\otimes n}||\mathcal{F}^{\otimes n})\leq \frac{1}{1-\epsilon}\limsup_{n\to\infty}\frac{1}{n}[D(\mathcal{E}^{\otimes n}||\mathcal{F}^{\otimes n}) + h(\varepsilon)] =  \frac{1}{1-\epsilon}D^{\text{reg}}(\mathcal{E}||\mathcal{F}) \\
  \implies  \lim_{\epsilon\to0}&\limsup_{n\to\infty}\frac{1}{n}D_H^\epsilon(\mathcal{E}^{\otimes n}||\mathcal{F}^{\otimes n}) \leq D^{\text{reg}}(\mathcal{E}||\mathcal{F}).
\end{align*}
\end{proof}

\section{Proof of Lemma \ref{lem:InfDImDmaxAEP}} \label{app:ProofOfDmaxAEPLemma}
\begin{proof}[Proof of \autoref{lem:InfDImDmaxAEP}]
To prove inequality \eqref{equ:DMaxAEPlowerbound}, we start with inequality \eqref{equ:ImpINeq3REP} from \autoref{lemma:SmoothedDmaxLowerBoundWW} for some $\alpha\in(0,1)$ applied to the states $\rho^{\otimes n}$ and $\sigma^{\otimes n}$, which, with the additivity of the $\alpha$-Petz-\renyi divergence yields
\begin{align*}
    \frac{1}{n}D^\epsilon_{\text{max}}(\rho^{\otimes n}||\sigma^{\otimes n}) \geq D_\alpha(\rho||\sigma)+\frac{1}{n}\frac{2}{\alpha-1}\log\left(\frac{1}{1-\epsilon}\right).
\end{align*}
Combining this with inequality \eqref{equ:PetzRenyiLowerContinuity} from Lemma \ref{lemma:Bjarne9} gives 
\begin{align*}
    \frac{1}{n}D^\epsilon_{\text{max}}(\rho^{\otimes n}||\sigma^{\otimes n}) \geq D(\rho||\sigma) -(1-\alpha)(c_\gamma(\rho||\sigma))^2 -\frac{1}{n}\frac{2}{1-\alpha}\log\left(\frac{1}{1-\epsilon}\right)
\end{align*}
under the constraint that $0<(1-\alpha)\leq\frac{\log3}{2c_\gamma(\rho||\sigma)}$. Choosing $1-\alpha=\frac{\log3}{2\sqrt{n}c_\gamma(\rho||\sigma)}$ satisfies this and gives
\begin{align*}
    \frac{1}{n}D^\epsilon_{\text{max}}(\rho^{\otimes n}||\sigma^{\otimes n}) &\geq D(\rho||\sigma)-\frac{c_\gamma(\rho||\sigma)}{\sqrt{n}}\left[\frac{\log3}{2}+\frac{4}{\log3}\log\left(\frac{1}{1-\epsilon}\right)\right] \\ &\geq D(\rho||\sigma)-\frac{4}{\sqrt{n}}c_\gamma(\rho||\sigma)\log\left(\frac{2}{1-\epsilon}\right),
\end{align*} where in the last inequality we used that $\log3>1$ and $\frac{\log3}{2}<4$.

To prove inequality \eqref{equ:DMaxAEPupperbound}, we start analogously with inequality \eqref{equ:ImpIneq2NEW} from \autoref{COR:2} for some $\alpha>1$. Applying this to $\rho^{\otimes n}, \sigma^{\otimes n}$, using the additivity of the Petz-Rényi divergence, and inequality \eqref{equ:PetzRenyiUpperContinuity} gives
\begin{align*}
    \frac{1}{n}D^\epsilon_\text{max}(\rho^{\otimes n}||\sigma^{\otimes n}) \leq D(\rho||\sigma)+\frac{1}{n}\frac{2}{\alpha-1}\log\left(\frac{1}{\epsilon}\right)+\frac{1}{n}\log\left(\frac{1}{1-\epsilon^2}\right)+(\alpha-1)(c_\gamma(\rho||\sigma))^2 
\end{align*} with the condition $0<\alpha-1\leq\frac{\gamma}{2}$. Picking again $\alpha-1=\frac{\log3}{2\sqrt{n}c_\gamma(\rho||\sigma)}$, satisfies the condition since $c_\gamma(\rho||\sigma)\geq\frac{\log3}{\gamma}$ and gives similarly to above 
\begin{align*}
   \frac{1}{n}D^\epsilon_\text{max}(\rho^{\otimes n}||\sigma^{\otimes n}) \leq D(\rho||\sigma) +\frac{4}{\sqrt{n}}c_\gamma(\rho||\sigma)\log\left(\frac{2}{\epsilon}\right)+\frac{1}{n}\log\left(\frac{1}{1-\epsilon^2}\right).
\end{align*}

\end{proof}

\section{Some more properties of the geometric \renyi divergence}
Subsequently we write $s(\rho)$ for the support of any operator $\rho$, and $\Pi_\rho$ for the orthogonal projection onto the support of $\rho$. 

\begin{proposition}[Alternative definition of the geometric \renyi trace function, \cite{hiai_quantum_2019}]\label{prop:geometric_renyi_alternative_definition}
    If $\rho, \sigma \in \P(\HS)$ are such that $\rho \leq c \sigma$ for some $0 < c < \infty$, then there exists a unique bounded operator $A$ with support and image in $s(\sigma)$, such that $\rho^{1/2} = A \sigma^{1/2}$, and this operator also satisfies $A^*A \leq c \Pi_\sigma$ and hence $\norm{A}_{\infty} \leq \sqrt{c}.$ If furthermore $b \sigma \leq \rho$ for some $0<b<\infty$ we additionally have $b \Pi_\sigma \leq A^*A $. We will subsequently also write $A = \rho^{1/2}\sigma^{-1/2}$ and $A^{*}A = \sigma^{-1/2}\rho\sigma^{-1/2}$. We can now define the geometric \renyi trace function between $\rho$ and $\sigma$ such that $b \sigma \leq \rho \leq c \sigma$ as
    \begin{equation}
        \Sg(\rho\|\sigma) = \Tr(\sigma (A^{*}A)^{\alpha}) = \Tr(\sigma (\sigma^{-1/2}\rho\sigma^{-1/2})^{\alpha})\,.
    \end{equation}
    which can be extended to general $\rho, \sigma$, via
    \begin{equation}
        \Sg(\rho\|\sigma) = \lim_{\varepsilon \to 0} \Sg(\rho + \varepsilon(\rho + \sigma)\| \sigma + \varepsilon(\rho + \sigma))\,.
    \end{equation}
    For $\alpha \in [0, 1) \cup (1,2]$ this definition agrees with the one given above in terms of minimal reverse tests.
\end{proposition}

\begin{lemma} \label{lem:geometric_anti_monotonicity}
If $\sigma_1 \leq \sigma_2$, and $\alpha \in [0,1)\cup (1,2]$, then 
\begin{equation}
    \Dg(\rho\|\sigma_1) \geq \Dg(\rho\|\sigma_2).
\end{equation}
\begin{proof}
    We start with the case $\alpha \in (1,2]$. 
    From \cite[Prop. 2.10]{hiai_quantum_2019}, for all $\rho, \sigma \in \P(\HS)$, we have
    \begin{equation}
        \Sg(\rho\|\sigma) = \widehat{S}_{1 - \alpha}(\sigma\|\rho).
    \end{equation}
    where $\widehat{S}_{1 - \alpha}(\sigma\|\rho)$ is defined as in \autoref{prop:geometric_renyi_alternative_definition}.
    If we now assume that there exists some $c < \infty$ such that $\rho \leq c\sigma_1 \leq c\sigma_2 \leq c^2 \rho$, then by \autoref{prop:geometric_renyi_alternative_definition} there exists a bounded $A_2$ such that $A_2 \rho^{1/2} = \sigma_2^{1/2}$, and 
    also a bounded $B$, with $B^*B \leq \mathbb{1}$, such that $B \sigma_2^{1/2} = \sigma_1^{1/2}$. Hence, $A_1 = B A_2$ satisfies $A_1 \rho^{1/2} = \sigma_1^{1/2}$, and we have $A_1^* A_1 = A_2^* B^*B A_2 \leq A_2^* A_2$. From \autoref{prop:geometric_renyi_alternative_definition} it also follows that $A_1^*A_1$ and $A_2^*A_2$ have bounded inverses on the support of $\rho$, and so all the following expressions are well-defined. It is well-known that $t \mapsto t^{p}$ is operator anti-monotone for $p \in [-1, 0]$, and thus,
    \begin{equation}
        \widehat{S}_{1 - \alpha}(\sigma_1\|\rho) = \Tr(\rho (A_1^* A_1)^{1 - \alpha}) \geq \Tr(\rho (A_2^* A_2)^{1 - \alpha}) = \widehat{S}_{1 - \alpha}(\sigma_2\|\rho)\,.
    \end{equation}
    The statement without requiring the existence of a $c < \infty$ then follows by taking limits as in \autoref{prop:geometric_renyi_alternative_definition}. For this, note that
    \begin{align}
        \Sg(\rho + \varepsilon(\rho + \sigma_2)\|\sigma_2 + \varepsilon(\rho + \sigma_2)) &\leq \Sg(\rho + \varepsilon(\rho + \sigma_2)\|\sigma_1 + \varepsilon(\rho + \sigma_2))\\ 
        &=\Sg(\rho + \varepsilon(\rho + \sigma_1) + \varepsilon(\sigma_2 - \sigma_1)\|\sigma_1 + \varepsilon(\rho + \sigma_1) + \varepsilon(\sigma_2 - \sigma_1)) \\
        &\leq \Sg(\rho + \varepsilon(\rho + \sigma_1)\|\sigma_1 + \varepsilon(\rho + \sigma_1))
    \end{align}
    where we used the generalized joint convexity \cite[Theorem 2.9]{hiai_quantum_2019} together with $\sigma_1\le\sigma_2$ in the last step.
    
    For $\alpha \in [0,1)$, $t \to t^{1-\alpha}$ is operator monotone, and with the analogous argument we get
    \begin{equation}
        \Sg(\rho + \varepsilon(\rho + \sigma_2)\|\sigma_1 + \varepsilon(\rho + \sigma_2)) \leq \Sg(\rho + \varepsilon(\rho + \sigma_2)\|\sigma_2 + \varepsilon(\rho + \sigma_2))\,.
    \end{equation}
    For taking limits, note that
    \begin{equation}
        \liminf_{\varepsilon \to 0} \Sg(\rho + \varepsilon(\rho + \sigma_2)\|\sigma_1 + \varepsilon(\rho + \sigma_2)) \geq \Sg(\rho\|\sigma_1)
    \end{equation}
    by the lower semi-continuity of $\Sg$ \cite[Theorem 5.5]{hiai_quantum_2019}. The desired statement then follows after taking logarithms and dividing by $\alpha - 1$ (which reverses the inequality). 
\end{proof}
\end{lemma}

In \cite[Proposition 2.11]{hiai_quantum_2019} it was shown that if $(s(\rho_1) \cup s(\sigma_1)) \perp (s(\rho_2) \cup s(\sigma_2))$ then $\Sg(\rho_1 + \rho_2\|\sigma_1 + \sigma_2) = \Sg(\rho_1\|\sigma_1) + \Sg(\rho_2\|\sigma_2)$. We require a similar statement for an infinite orthogonal decomposition, which is not directly implied, so we give a proof here. The argument is still essentially equivalent to the one given for a finite orthogonal decomposition. 
\begin{lemma}\label{lem:generalized_direct_sum_decomp}
    For $i \in \naturals$, let $\rho_i, \sigma_i \in \cB_1$ be such that $(s(\rho_i) \cup s(\sigma_i)) \perp (s(\rho_j) \cup s(\sigma_j))$ for all $i,j$. Then, for $\alpha \in [0, 1), (1,2]$:
    \begin{equation}
        \Sg\bigg(\sum_i \rho_i\bigg\|\sum_i \sigma_i\bigg) = \sum_i \Sg(\rho_i\|\sigma_i)\,.
    \end{equation}
\end{lemma}
\begin{proof}
In light of \autoref{prop:geometric_renyi_alternative_definition} we can show the statement in the case where there exists a $c < \infty$ such that $\rho_i \leq c \sigma_i$ for all $i$, and the full statement follows by taking limits. For each $i$ pick $A_i$ such that $\rho_i^{1/2} = A_i \sigma_i^{1/2}$. By \autoref{prop:geometric_renyi_alternative_definition} we can pick such an $A_i$ with image and support contained in $s(\sigma_i)$ which are all orthogonal, and so we get the following three statements
\begin{align}
    \bigg(\sum_i \rho_i\bigg)^{1/2} &= \bigg(\sum_i \rho_i^{1/2}\bigg) = \bigg(\sum_i A_i\bigg) \bigg(\sum_i \sigma_i^{1/2}\bigg) = \bigg(\sum_i A_i\bigg) \bigg(\sum_i \sigma_i\bigg)^{1/2},\\
    \Bigg(\bigg(\sum_i A_i^{*}\bigg)\bigg(\sum_j A_j\bigg)\Bigg)^{\alpha} &= \sum_i (A_i^* A_i)^{\alpha},\\
    \sigma \sum_i (A_i^* A_i)^{\alpha} &= \sum_i \sigma_i (A_i^* A_i)^{\alpha},
\end{align}
which imply the desired equality.
\end{proof}
\section{Some more properties of the quantum relative entropy}

\begin{lemma}\label{lem:entropy_almost_concave}
    Let $\sigma \in \DM$ be a density matrix, $\{\lambda_i\}$, $i = 1, ..., n$ be a normalized probability distribution, and $\rho_i \in \DM$, $i = 1, ..., n$ be a set of density matrices. Then,
    \begin{equation}
         D\left(\sum_{i=1}^n \lambda_i \rho_i\Bigg\|\sigma\right) + H(\lambda) \geq \sum_{i=1}^n \lambda_i D(\rho_i\|\sigma) 
    \end{equation}
    where $H(\lambda)$ is the Shannon entropy of $\lambda$. 
	\end{lemma}
	\begin{proof}
		In \cite{Lindblad_RelativeEntropy_1973} it was shown that for normalized states $\rho, \sigma$ the relative entropy in infinite dimensions can be written as
		\begin{align}
		    D(\rho\|\sigma) &= \sum_{i,j}|\langle e_i, f_j \rangle|^2 (p_i \log p_i - p_i\log q_j - p_i + q_i) \\
		    &= \sum_j \bra{e_j} \rho \log \rho - \rho \log \sigma - \rho + \sigma \ket{e_j} \\ 
		    &= \sum_j \bra{f_j} \rho \log \rho - \rho \log \sigma - \rho + \sigma \ket{f_j},
		\end{align}
        where $\rho = \sum_i p_i \kb{e_j}$ and $\sigma = \sum_j q_j \kb{f_j}$ are the spectral decompositions of $\rho$ and $\sigma$, and it is easy to see (using $x \geq 1 + \log(x)$) that the expressions in the sum are always positive, and hence the sum is well-defined, although possibly infinite. 
		
		We write $\sigma = \sum_j q_j \kb{f_j}$ as above, $\rho = \sum_{i=1}^n \lambda_i \rho_i$, and note that $\rho \geq \lambda_i \rho_i$ for all $i$. We then have
		\begin{align}
		    \sum_{i=1}^n\lambda_i D(\rho_i\|\sigma) &= \sum_{i=1}^n \lambda_i \sum_j \bra{f_j} \rho_i \log \rho_i - \rho_i \log \sigma - \rho_i + \sigma \ket{f_j} \\
		    &= \sum_j \bra{f_j} \sum_{i=1}^n \lambda_i \rho_i \log\rho_i - \rho \log \sigma - \rho + \sigma \ket{f_j} \\ 
		    &= \sum_j \bra{f_j} \sum_{i=1}^n \lambda_i \rho_i \log(\lambda_i \rho_i) - \sum_{i=1}^n \lambda_i \rho_i \log(\lambda_i) - \rho \log \sigma - \rho + \sigma \ket{f_j} \\
		    &\leq \sum_j \bra{f_j} \rho \log \rho - \rho \log \sigma - \rho + \sigma - \sum_{i=1}^n \lambda_i \rho_i \log(p_i) \ket{f_j} \\
		    &= \sum_j \bra{f_j} \rho \log \rho - \rho \log \sigma - \rho + \sigma \ket{f_j} + 
		    \sum_j \sum_{i=1}^n (-\lambda_i\log(\lambda_i)) \bra{f_j}\rho_i \ket{f_j}\\
		    &= D(\rho\|\sigma) + H(\lambda)
		\end{align}
		where in the inequality step we used the operator monotonicity of the logarithm. Also, based on the argument mentioned above, one easily checks that in all sums over $j$ the summands are always positive, and so all the expressions are well-defined.

	\end{proof}

\end{document}